\documentclass[11pt]{amsart}
\usepackage[foot]{amsaddr}
\usepackage{color}
\usepackage{amsmath}
\usepackage{amsfonts}
\usepackage{mathtools}

\usepackage{amssymb}
\newtheorem{thm}{Theorem}[section]

\newtheorem{cor}[thm]{Corollary}

\newtheorem{lemma}[thm]{Lemma}
\newtheorem{prop}[thm]{Proposition}
\newtheorem{prob}[thm]{Problem}
\theoremstyle{definition}
\newtheorem{defn}[thm]{Definition}
\newtheorem{remark}[thm]{Remark}
\newtheorem{exam}[thm]{Example}

\DeclareMathOperator{\diag}{diag}
\DeclareMathOperator{\Diag}{Diag}

\newcommand{\ip}[2]{\langle #1 , #2 \rangle}

\newcommand{\bb}[1]{\mathbb{#1}}
\newcommand{\cl}[1]{\mathcal{#1}}

\begin{document}

\title[]{Synchronous Values of Games}

\author[J.~W.~Helton]{J. William Helton}
\address[J.~W.~Helton]{Department of Mathematics, University of California San Diego, La Jolla, CA 92093-0112}
\email{helton@math.ucsd.edu}
\author[H.~Mousavi]{Hamoon Mousavi}
\address[H.~Mousavi]{Department of Computer Science, Columbia University, New York, NY 10027}
\email{sm5053@columbia.edu}
\author[S.~S.~Nezhadi]{Seyed Sajjad Nezhadi}
\address[S.~S.~Nezhadi]{Joint Center for Quantum Information and Computer Science and Department of Computer Science, University of Maryland, College Park, MD 20742 }
\email{sajjad@umd.edu}
\author[V.~I.~Paulsen]{Vern I.~Paulsen}
\address[V.~I.~Paulsen]{Institute for Quantum Computing and Department of Pure Mathematics, University of Waterloo,
Waterloo, ON, Canada  N2L 3G1}
\email{vpaulsen@uwaterloo.ca}
\author[T.~B.~Russell]{Travis B.~Russell}
\address[T.~B.~Russell]{Department of Mathematics, Texas Christian University, Fort Worth, TX, 76129}
\email{travis.b.russell@tcu.edu}

\begin{abstract} We study synchronous values of games, especially synchronous games.  It is known that a synchronous game has a perfect strategy if and only if it has a perfect synchronous strategy. However, we give examples of synchronous games, in particular graph colouring games, with synchronous value that is strictly smaller than their ordinary value. Thus, the optimal strategy for a synchronous game need not be synchronous.

We derive a formula for the synchronous value of an XOR game as an optimization problem over a spectrahedron involving a matrix related to the cost matrix.

We give an example of a game such that the synchronous value of repeated products of the game is strictly increasing. We show that the synchronous quantum bias of the XOR of two XOR games is not multiplicative.

Finally, we derive geometric and algebraic conditions that a set of projections that yields the synchronous value of a game must satisfy.
\end{abstract}

\maketitle

\subsection*{Acknowledgments} H.M. acknowledges the support of the Natural Sciences and Engineering Research Council of Canada (NSERC). V.I.P. was supported by NSERC grant 03784. All the authors wish to thank the American Institute of Mathematics (AIM) where this research originated. 

\newpage

\tableofcontents
\section{Introduction}  Nonlocal games have been the central object of study in many areas of computer science and quantum information \cite{BGKW, BFL, FGLSS, ALMSS, AS}. They play a central role in our understanding of entanglement.  Such games were vital to the recent resolution of the Connes' Embedding Problem \cite{JNVWY} and to answering the Tsirelson Problems \cite{Sla, JNVWY} about the relationships between the different mathematical models for entanglement.

The {\it value} of a nonlocal game is the supremum of the probability of winning the game over all allowed strategies. The value of a game can vary depending on the types of strategies or probability densities that are allowed, and there has been considerable interest in how the value of a game can change when one is allowed to use quantum assisted strategies versus classically defined distributions \cite{CHSH, Me, CHTW, RV, CM, CLS}.
 In addition, the proofs of the separation of the various mathematical models for entanglement involved finding games whose quantum assisted values depended on the particular mathematical model used to describe entanglement. Thus, separating the values of games for the various models has been the most successful tool in showing that these various models of quantum densities are different \cite{JNVWY,Slb,Sla,DPP,MR,Co,CS}.
 
 In this paper we are interested in how values of games behave when one puts on the restriction that the probability densities derived from the various models must also be {\it synchronous}, a term we define later. There are several reasons for this interest. First, it has been shown that the study of synchronous densities is related to the study of traces on C*-algebras \cite{PSSTW, KPS}. For this reason, finding synchronous values of games turns into problems about optimizing the trace of an element of a C*-algebra over certain types of traces on the C*-algebra, which lends a totally different flavour to the theory of values of games.
 
 Second, the Connes' Embedding Problem in its original form is a question about the behaviour of traces. So studying synchronous values of games provides a much more direct link between this problem and games.
 
 Finally, there is a family of games known as {\it synchronous games} that has been very useful in delineating the separations between the different models for quantum densities. In fact, the separations between the different models for entanglement have all been shown using synchronous games. For synchronous games, it is very natural to restrict the allowed strategies to also be synchronous. 
 
 Thus, hopefully, the study of synchronous values of synchronous games could lead to a clearer understanding of the negative resolution of the Connes' Embedding Problem. 
 
 In section 2, we delineate these ideas and definitions more clearly. 
 
 In section 3, we turn our attention to the graph colouring game. In this game the players are given $c$ colours with $c$ smaller than the chromatic number of the graph. The value of this game is in some sense a measure of how nearly they can convince someone that they have successfully coloured the graph with only $c$-colours. Remarkably the quantum assisted value can be much higher than the classical value of these games. 
 
 We show that for a particular density on inputs, the synchronous local value of this game is a function of the max c-cut of the graph, while the ordinary local value is related to the max cut problem for a bipartite extension of the graph. This leads us to introduce a quantum version of max cut that is motivated by the quantum assisted synchronous value of the 2-colouring game and we prove that this value is given by an SDP. There are many SDP relaxations of max cut, and our results show that one of these relaxations corresponds to the synchronous value of this game. For an introduction to this literature see \cite{La}.
 We give a formula for the quantum assisted synchronous value of the 2-colouring game of a graph with any density on inputs in terms of an SDP and compute this value for some graphs. 
 
 In section 4, we turn our attention to a family of games that has been studied extensively in the literature, called XOR games. For XOR games, their ordinary value and their synchronous value are shown to be optimization problems over two different spectrahedra. 
 
 In section 5, we return to the graph colouring game and study the problem of 2-colouring an odd cycle. Even though this game is synchronous, we show that often there are non-synchronous strategies that out perform any synchronous strategy. In fact, we show that as one varies the prior distributions on pairs of vertices, which are the inputs of this game, then there are various regimes where the synchronous values are smaller than the non-synchronous values and other regimes where they are the same.
 
 In section 6, we turn our attention to parallel repetition of games. A famous result in game theory says that unless the classical value of a game is 1, then the value of playing $n$ parallel copies of the game tends to 0 as $n$ grows. In contrast, we give an example of a game whose synchronous value is strictly  increasing under parallel repetition. 
 The {\it bias} of the XOR of two XOR games is known to be multiplicative. We show that in contrast the synchronous bias need not be multiplicative.
 
 Finally, each synchronous strategy  for a game corresponds to a certain arrangement of projections in a tracial C*-algebra. In section 7, we derive conditions that are necessarily met by any arrangement of projections that yield a correlation that attains the synchronous value of the game. For the CHSH game we show that these relations force all of the projections to commute, and that, consequently, for the CHSH game the quantum-assisted synchronous value is equal to the classical value of the game. More generally, we give  conditions which must hold
whenever the max value of a game occurs with a finite dimensional synchronous strategy.

 \section{Values of Games}
 
 The types of games that we shall be interested in are {\it  two player nonlocal games}. These are {\it cooperative games} in which two players referred to as Alice and Bob cooperate to give correct pairs of answers to pairs of questions posed by a third party often called the {\it Referee} or {\it Verifier}. The nonlocality condition is that once the game starts the players cannot communicate with one another. In particular, Alice does not know what question Bob has received and vice versa. 
 Whether the pair of answers returned by the players is satisfactory or not depends not just on the individual answers but on the 4-tuple consisting of the question-answer pairs.  
 
More formally a nonlocal game is described by two input sets $I_A, I_B$, two output set $O_A, O_B$,
 and a function
 \[ \lambda: I_A \times I_B \times O_A \times O_B \to \{ 0, 1\},\]
 often called the {\it rules} or {\it verification function},
 where 
 \[W:= \{ (x,y,a,b) : \lambda(x,y,a,b) = 1 \},\]
 is the set of {\it correct} or {\it winning} 4-tuples and
 \[ N:= \{ (x,y,a,b): \lambda(x,y,a,b) =0 \},\]
 is the set of {\it incorrect} or {\it losing} 4-tuples. We sometimes refer to $N$ as the \emph{null set}.
 Each {\it round} of the game consists of Alice and Bob receiving an input pair $(x,y)$ and returning an output pair $(a,b)$. Thus, a game $G$ is specified by $(I_A, I_B, O_A, O_b, \lambda)$.
 
 Intuitively, if Alice and Bob have some {\it strategy} for such a game, then it would yield a {\bf conditional probability density}\footnote{Some authors refer to conditional probability densities as correlations.},
 \[ p(a,b|x,y), \,\, x \in I_A, \, y \in I_B, \, a \in O_A, \, b \in O_B,\]
 which gives the conditional probability that Alice and Bob return output pair $(a,b)$, given that they received input pair $(x,y)$. 
 
 A {\it deterministic strategy} corresponds to a pair of functions, $f_A:I_A \to O_A$ and $f_B:I_B \to O_B$ such that any time Alice and Bob receive input pair $(x,y)$ they reply with output pair $(a,b)=(f_A(x), f_B(y))$.
 In this case $p(a,b|x,y)$ is always $0$ or $1$.
 
 We often use {\bf density} to refer to conditional probability density. We generally identify strategies with the conditional densities that they produce.
 Since $0 \le p(a,b|x,y) \le 1, \, \forall a,b,x,y$, it is natural to identify densities with points in the $m$-cube, $[0,1]^m$ where $m=n_A n_B k_A k_B$ is the product of the cardinalities, $n_A= |I_A|, \,  n_B= |O_B|,  \, k_A= |O_A|, \,  k_B= |O_B|$.
 
A strategy $p$ is called \emph{non-signaling} if 
\begin{itemize}
    \item for every $a \in O_A, x\in I_A$ and $y,y' \in I_B$ we have \[\sum_b p(a,b|x,y) = \sum_b p(a,b|x,y'),\]
    \item for every $b \in O_B,y\in I_B$ and $x,x' \in I_A$ we have  \[\sum_a p(a,b|x,y) = \sum_a p(a,b|x',y).\]
\end{itemize}
Intuitively, this is a restatement of the nonlocality condition that states that Alice's answer is not dependent on Bob's question and vice versa. Every strategy in this paper is non-signalling. For a density $p(a,b|x,y)$ we denote Alice's marginal density by $p_A$. This is defined to be $p_A(a|x) = \sum_b p_A(a,b|x,y)$ where $y$ is any question for Bob (the choice of $y$ does not matter because $p$ is non-signalling). One can similarly define Bob's marginal density $p_B$.
 
 In a two-player nonlocal game, we sometimes specify that the referee asks Alice and Bob questions according to a given {\bf prior distribution}\footnote{Some authors let $\pi$ be a part of the definition of the game, that is they let the tuple $(I_A, I_B, O_A, O_b, \lambda, \pi)$ to be specifying the game.} (or distribution for short) on input pairs, i.e.,
 \[ \pi:I_A \times I_B \to [0,1],\]
 with $\sum_{x,y} \pi(x,y) = 1$. Then the probability of winning, i.e., the {\it expected value} of a given strategy $p(a,b|x,y)$ is given by
 \begin{eqnarray}
 \omega(G, \pi, p) & = & \sum_{x,y,a,b} \pi(x,y) \lambda(x,y,a,b) p(a,b|x,y) \nonumber \\
 & = & \sum_{(x,y,a,b) \in W} \pi(x,y) p(a,b|x,y). \nonumber
 \end{eqnarray}
 Given a set $S$ of conditional probability densities the  {\bf $S$-value} of the 
 pair $(G, \pi)$ is
 \[ \omega_S(G, \pi) := \sup \{ \omega(G,\pi,p): p \in S \}.\]
 Identifying $S \subseteq [0,1]^m$, since the value is clearly a convex function of $p$, the value will always be attained at one of the extreme points of the closed convex hull of $S$. 
 
 There are many sets of conditional probability densities for which researchers attempt to compute the $S$-value. Among these, in particular, are the {\bf local, quantum}, and {\bf quantum commuting} densities, denoted by 
 \[ C_{loc}(n_A,n_B, k_A,k_B), C_{q}(n_A, n_B, k_A, k_B), \text{ and } C_{qc}(n_A, n_B, k_A, k_B), \] respectively.   We refer to \cite{HMPS, KPS} for the precise definitions of these sets. To simplify notation, we generally suppress the set sizes. For fixed numbers of inputs and outputs these are convex sets, with $C_{loc}$ and $C_{qc}$ closed, while $C_q$ is not generally closed. In fact, in \cite{DPP} it was shown that $C_q(n,n,k,k)$ is not closed for all $n \ge 5, k \ge 2$. The closure of $C_q$ is often denoted by $C_{qa}$. These sets satisfy
 \[ C_{loc} \subseteq C_{q} \subseteq C_{qa} \subseteq C_{qc}.\]
 
 We remark that $C_{loc}$ is a convex polytope whose extreme points are generated by the $\{0,1 \}$ densities arising from deterministic strategies.

 To simplify notation, we set
 \[ \omega_t(G, \pi) = \omega_{C_t}(G, \pi), \,\, t= loc, q, qa, qc.\]
 Note that, since the value is a convex function of the densitiy, we have that 
 \[ \omega_{loc}(G, \pi) = \sup \{ \sum_{\stackrel{x,y}{(x,y,f_A(x), f_B(y)) \in W}} \pi(x,y) \},\]
 where the supremum is over all pairs of functions $f_A:I_A \to O_A, \,\, f_B: I_B \to O_B$.

 Also, since the value is a continuous function of the density, we have $\omega_{q}(G, \pi) = \omega_{qa}(G, \pi)$. An often interesting question for $\omega_q(G, \pi)$ is whether or not the value is actually attained by an element of $C_q$.  For $t= loc, qa,qc$ the value is always attained, since the corresponding sets of densities are closed and hence compact.
 
 Computing these values for various games generated a great deal of interest in the operator algebras community when it was shown by \cite{JNPPSW} that if the {\it Connes' embedding problem} had an affirmative answer, then
 \[ \omega_q(G, \pi) = \omega_{qc}(G, \pi),\]
 for all games and densities.
 
 Recently, \cite{JNVWY} proved the existence of a game for which
 \[ \omega_q(G, \pi) < 1/2 < \omega_{qc}(G, \pi) =1,\]
 thus refuting the embedding problem.

 \subsection{Synchronous Games}
 
 The games that we shall be interested in have the property that Alice and Bob's question sets and answer sets are the same, i.e., $I_A= I_B \eqqcolon I$ and $O_A= O_B \eqqcolon O$. So such a game is given as
 $G= (I, O, \lambda)$. If $n=|I|$ and $k= |O|$, then we say that the game has  $n$ inputs and $k$ outputs and write $C_t(n,k), \, t= loc, q, qc$ for the corresponding sets of densities.
 
 For such games it is natural to impose some conditions on $\lambda$.  We call $G$ {\bf synchronous} if
 \[ \lambda(x,x,a,b) = 0, \, \forall a \ne b,\]
 i.e., if Alice and Bob are asked the same question they must give the same reply, although their answer to this question could vary with rounds. The game constructed in \cite{JNVWY} that refutes the embedding problem is synchronous. 
 
  We call a game {\bf symmetric}
 if
 \[ \lambda(x,y,a,b) = \lambda(y,x,b,a),\]
 so that interchanging Alice and Bob has no effect on the rules.
 
 In addition to imposing these conditions on the rules of a game, it is natural to impose them on the allowed densities.  A density $p(a,b|x,y)$ is called {\bf synchronous} if
 \[ p(a,b|x,x) =0, \, \forall a \ne b.\]
 We let $C^s_t(n,k) \subseteq C_t(n,k), \,  t = loc, q, qc$ denote the corresponding subsets of synchronous densities.

 Given a game $G=(I,O, \lambda)$ with distribution $\pi$ we set
 \[ \omega^s_t(G, \pi) = \omega_{C_t^s}(G, \pi), \, t= loc, q, qc.\]
 These are the values that we are interested in computing in this paper.

 In \cite{PSSTW}, which introduced the concept of synchronous games and densities, and \cite{KPS} each of the sets $C^s_t(n,k), t= loc, q,qa, qc$ were characterized in terms of traces.

 Given a C*-algebra $\cl A$ with unit, by a {\bf trace} on $\cl A$ we mean a linear functional $\tau: \cl A \to \bb C$ satisfying $\tau(I) = 1$, $p \ge 0 \implies \tau(p) \ge 0$ and $\tau(xy) = \tau(yx)$.  The first two conditions characterize {\bf states} on $\cl A$. When $\cl A= M_n$, the set of $n \times n$ matrices, it is known that there is a unique trace, namely,
 \[ tr_n((a_{i,j})) = \frac{1}{n} \sum_i a_{i,i} = \frac{1}{n} Tr((a_{i,j})).\]
 
 Given a C*-algebra $\cl A$ with unit $I$, a {\bf k-outcome projection valued measure(k-PVM)} is a set of $k$ projections, $E_a = E_a^2 = E_a^*$ such that $\sum_{a=0}^{k-1} E_a = I$. A family of $n$ k-PVM's is a set of projections $\{ E_{x,a}: 1 \le x \le n, 0 \le a \le k-1 \}$ with 
$\sum_a E_{x,a} = I, \forall x$.

The following is a restatement of the results of \cite{PSSTW} and \cite{KPS} characterizing elements of $C^s_t(n,k)$ in terms of traces.

\begin{thm}[\cite{PSSTW, KPS}] We have that $p \in C^s_{qc}(n,k)$ if and only if there is a family of $n$ $k$-outcome PVM's  $\{ E_{x,a}: 1 \le x \le n, 0 \le a \le k - 1 \}$ in a C*-algebra $\cl A$ with a trace $\tau$ such that
\[ p(a,b|x,y) = \tau(E_{x,a}E_{y,b}).\]
Moreover,
\begin{itemize}
\item $p \in C^s_{loc}(n,k)$ if and only if $\cl A$ can be taken to be abelian, 
\item $p \in C^s_q(n,k)$ if and only if $\cl A$ can be taken to be finite dimensional, 
\item $p \in C^s_{qa}(n,k)$ if and only if $\cl A$ can be taken to be an ultrapower of the hyperfinite $II_1$-factor.
\end{itemize}
\end{thm}

Note that if $p(a,b|x,y)$ is a synchronous density, then
\[ p(a,b|x,y) = \tau(E_{x,a}E_{y,b}) = \tau(E_{y,b}E_{x,a}) = p(b,a|y,x).\]
In other words every synchronous density is {\bf symmetric}.

The above result translates into the following result about synchronous values.

\begin{thm} Let $G= (I,O,\lambda)$ be an n input k output game and let $\pi$ be a prior distribution on inputs.  Then
\begin{enumerate}
\item
\[ \omega^s_{loc}(G, \pi) = \sup \{ \sum_{\stackrel{x,y}{(x,y,f(x),f(y)) \in W}} \pi(x,y) \},\]
where the supremum is over all functions, $f:I \to O$ from inputs to outputs,
\item
\[ \omega^s_q(G, \pi) = \omega^s_{qa}(G, \pi) = \sup \{ \sum_{(x,y,a,b) \in W} \pi(x,y) tr_m(E_{x,a}E_{y,b}) \},\]
where the supremum is over all families of $n$ k-PVM's in $M_m$ and over all $m$,
\item
\[ \omega^s_{qc}(G, \pi) = \sup \{ \sum_{(x,y,a,b) \in W} \pi(x,y) \tau(E_{x,a}E_{y,b}) \},\]
where the supremum is over all unital C*-algebras $\cl A$, traces $\tau$, and families of $n$ k-PVM's in $\cl A$.
\end{enumerate}
\end{thm}

As we remarked earlier, the second supremum may not be attained.

\subsection{A Universal C*-algebra Viewpoint}
We let $\bb F(n,k)$ denote the group that is the free product of $n$ copies of the cyclic group of order $k$. The full C*-algebra of this group $C^*(\bb F(n,k))$ is generated by $n$ unitaries $u_x, 1 \le x \le n$ each of order $k$, i.e., $u_x^k = I$.
Given any unital C*-algebra $\cl A$ with $n$ unitaries $U_x \in \cl A, 1 \le x \le n$ of order $k$, there is a *-homomorphism from $C^*(\bb F(n,k))$ mapping $u_x \to U_x$. If we decompose each $u_x$ in terms of its spectral projections,
\[ u_x = \sum_{a=0}^{k-1} \alpha^a e_{x,a}, \]
where $\alpha= e^{2 \pi i/k}$, 
then $\{ e_{x,a}: 1 \le x \le n, 0 \le a \le k-1 \}$ is a universal family of $n$ k-PVM's, in the sense that given any set of $n$ k-PVM's $\{ E_{x,a} \}$ in a unital C*-algebra $\cl A$, there is a unital *-homomorphism from $C^*(\bb F(n,k))$  to $\cl A$ sending $e_{x,a} \to E_{x,a}$.

Values of games can be interpreted in terms of properties of the maximal and minimal C*-tensor product of this algebra with itself. 

It follows from the work of \cite{JNPPSW}(see also \cite{PT}) that
\begin{itemize}
    \item $p(a,b|x,y) \in C_q(n,k)^-= C_{qa}(n,k)$ if and only if there exits a state \[s:C^*(\bb F(n,k)) \otimes_{min} C^*(\bb F(n,k)) \to \bb C\] such that \[p(a,b|x,y) = s(e_{x,a} \otimes e_{y,b}),\]
    \item $p(a,b|x,y) \in C_{qc}(n,k)$ if and only if there exists a state \[s: C^*(\bb F(n,k)) \otimes_{max} C^*(\bb F(n,k)) \to \bb C\] such that \[p(a,b|x,y) = s(e_{x,a} \otimes e_{y,b}).\]
\end{itemize}

Given a game $G$ and prior distribution $\pi$ we set
\[ P_{G,\pi} = \sum_{(x,y,a,b) \in W} \pi(x,y) e_{x,a} \otimes e_{y,b}.\]
Using the fact that norms of positive elements are attained by taking the supremum over states, we have:

\begin{prop} Given an n input, k output game $G= (I,O, \lambda)$ with distribution $\pi$,
\[ \omega_q(G, \pi) = \|P_{G, \pi} \|_{C^*(\bb F(n,k)) \otimes_{min} C^*(\bb F(n,k))},\]
and
\[ \omega_{qc}(G, \pi) = \|P_{G, \pi} \|_{C^*(\bb F(n,k) \otimes_{max} C^*(\bb F(n,k))}.\]
\end{prop}

The example of \cite{JNVWY} gave the first proof that the minimal and maximal norms are different.

We now turn to the synchronous case. The element $e_{x,a} e_{y,b}$ is not positive, but for any trace we have that
\[ \tau(e_{x,a} e_{y,b}) = \tau(e_{x,a} e_{y,b} e_{x,a}),\]
and $e_{x,a}e_{y,b}e_{x,a} \ge 0$.

We set
\[ R_{G, \pi} = \sum_{(x,y,a,b) \in W} \pi(x,y) e_{x,a} e_{y,b} e_{x,a}.\]

We also set $\cl C \subseteq C^*(\bb F(n,k))$ equal to the closed linear span of all commutators, $\{x,y\} = xy - yx$.

Given any C*-algebra $\cl A$ we let $T(\cl A)$ denote the set of traces on $\cl A$ and let $T_{fin}(\cl A)$ denote the set of traces that {\it factor through matrix algebras,} i.e., are of the form
\[ \tau(a) = tr_m(\pi(a)),\]
for some $m$ and some unital *-homomorphism $\pi: \cl A \to M_m$.

\begin{thm} Let $G=(I,O, \lambda)$ be an n input, k output game with distribution $\pi$. Then
\begin{enumerate}
\item \begin{align*}\omega^s_{qc}(G, \pi) &= \sup \{ \tau(R_{G,\pi}): \tau \in T(C^*(\bb F(n,k)) \} \\
  &=  \inf \{ \|R_{G,\pi} - C \|: C \in \cl C \},\end{align*}
\item \begin{align*}\omega^s_q(G, \pi) = \sup \{ \tau(R_{G,\pi}) : \tau \in T_{fin}(C^*(\bb F(n,k)) \}.\end{align*}
\end{enumerate}
\end{thm}

Two of the equalities are direct applications of the above facts. 
The equality of the value with the distance to the space of commutators follows from \cite[Theorem~2.9]{CP} where it is shown that for positive elements of a C*-algebra, the supremum over all traces is equal to the distance to the space $\cl C$. 

For the example of a game constructed in \cite{JNVWY}, it is known that
\[ \omega^s_q(G, \pi) < 1/2 < \omega^s_{qc}(G, \pi) =1, \]
and consequently, their results also give the first proof that  $T_{fin}(C^*(\bb F(n,k))$ is not dense in $T(C^*(\bb F(n,k))$. Perhaps even more remarkable is that this difference is witnessed by the element $R_{G, \pi}$ for some game, which only involves words in the generators of order three. However, the game of \cite{JNVWY} is mostly given implicitly and estimates on the values of $n$ and $k$ to achieve their example are very large.

 In summary, we see that the theory of values and synchronous values of these games gives us interesting information about C*-algebras.  Thus, we are led to study these values for interesting sets of games.

\section{The Graph Colouring Game} 
In this section we study the synchronous value of the game we get by trying to colour the vertices of a graph using $c$-colours, especially when $c$ is smaller than the least number of colours needed for an actual colouring. By a graph we mean a pair $G=(V, E)$, where $V$ denotes the vertices and $E \subseteq V \times V$ denotes the edge set. Our graphs are undirected, i.e., $(x,y) \in E \implies (y, x) \in E$ and loopless, i.e., $(x,x) \notin E$. A {\it $c$-colouring} is any function $f:V \to \{ 1,..., c \}$ such that $(x,y) \in E$ implies that $f(x) \ne f(y)$.


  Note that since $(x,y) \in E \implies (y,x) \in E$, and these both represent the same edge, then the cardinality of the set $E$ is equal to twice the number of edges.

   Before recalling the graph colouring game it helps to recall the {\bf graph homomorphism game}.

Given two graphs $G_i=(V_i, E_i)$ a {\it graph homomorphism} is a function $f:V_1 \to V_2$ such that $(x,y) \in E_1 \implies (f(x), f(y)) \in E_2$. If we let $K_c$ denote the complete graph on $c$ vertices, then a $c$-colouring of $G$ is just a graph homomorphism from $G$ to $K_c$.

The {\bf graph homomorphism game, $Hom(G_1, G_2)$} is the synchronous game with inputs $I= V_1$, outputs $O= V_2$ and rule $\lambda: V_1 \times V_1 \times V_2 \times V_2 \to \{ 0,1 \}$ with null set
\[ N = \{(x,y,a,b): (x,y) \in E_1, \, (a,b) \notin E_2 \} \cup \{ (x,x,a,b): x \in V_1, \, a \ne b \}.\]
Note that $\lambda$ is symmetric.

The {\bf graph $c$-colouring game} is the game $Hom(G, K_c)$. We use $\{ 1,...,c \}$ for the vertex set of $K_c$. We also usually assume that $c < \chi(G)$ (where $\chi(G)$ is the chromatic number of $G$) since otherwise
\[ \omega_t^s(Hom(G, K_c)) =  1, \text{ for } t= loc,q,qa,qc.\]


\subsection{The Relation Between Max c-Cut and the Synchronous Local Value}
Given a graph $G=(V, E)$ the {\it max c-cut} of $G$, is the maximum number of edges that can be coloured ``correctly" using $c$-colours, i.e.,
\[Cut_c(G):= \frac{\max \{ | \{ (x,y) \in E: x \in S_i, \, y \in S_j, \, i \ne j \} | \}}{2}, \]
where the maximum is over all partitions of $V$ into $c$ disjoint subsets, $S_1,..., S_c$ and the absolute value signs denote cardinality. Equivalently, a partition into $c$ disjoint subsets is defined by a function $f:V \to \{ 1,..., c \}$ with $S_i = f^{-1}(\{ i \}),$ so that
\[ Cut_c(G) = \frac{\max \{ | \{ (x,y) \in E: f(x) \ne f(y) \}| \}}{2},\]
where now the maximum is over all functions.  Note that $G$ has a c-colouring precisely when $\frac{|E|}{2} = Cut_c(G)$. 

The max 2-cut is generally referred to as simply the max cut. Computing the max cut is known to be NP-hard \cite{Ka}.

The following result shows that from the point of view of max cut problems, the synchronous value of the graph colouring game is more meaningful.

\begin{prop} Let $G=(V,E)$ be a graph on $n$ vertices and let $Hom(G, K_c)$ be the graph c-colouring game and let $\pi$ be the uniform density on $E$. Then
\[ \omega_{loc}^s(Hom(G, K_c), \pi) = \frac{2Cut_c(G)}{|E|}.\]
\end{prop}
\begin{proof} Each synchronous deterministic strategy corresponds to a function $f: V \to \{ 1,...,c \}$. The number of input pairs for which this strategy will win is equal to
$2Cut_c(G)$ and the result follows.
\end{proof}

In contrast, one can see that $\omega_{loc}(Hom(G, K_c))$ is related to the max c-cut of a bipartite graph over $G$, since Alice and Bob are allowed different functions for their deterministic strategy. Given a graph $G= (V,E)$ we define a new graph $G_b= (V_b, E_b)$ with
$V_b = V \times \{ 0, 1 \}$ and $((x,i), (y,j)) \in E_b$ if and only if
$i \ne j$ and $(x,y) \in E$. This graph is the usual bipartite graph defined over $G$.

\begin{prop} Let $G=(V, E)$ be a graph, let $G_b= (V_b, E_b)$ be the bipartite graph defined over $G$ as above, and consider the c-colouring game with  $\pi$ the uniform probability density on $E$. Then
\[ \omega_{loc}(Hom(G, K_c), \pi) = \frac{Cut_c(G_b)}{|E|}.\]
\end{prop}
\begin{proof} Each deterministic strategy is given by a pair of functions $f,g: V \to \{ 1,..., c \}$. 
Such pairs of functions are in one-to-one correspondence with functions $F: V_b \to \{ 1,..., c \}$ by setting $f(x) = F((x,0))$ and $g(x) = F((x,1))$.

The number of times that this strategy will win is equal to
\[ | \{ (x,y) \in E: f(x) \ne g(y) \}|= | \{ ((x,0),(y,1)) \in E_b : F(x,0) \ne F(y,1) \} |.\] 
Note that when we chose $f,g$ to maximize this number, we are obtaining $Cut_c(G_b)$ the actual number of edges since we are not counting ordered pairs of the form $((x,1),(y,0))$,
and the result follows.
\end{proof} 

Thus, there is a clean relationship between the synchronous local value of the graph colouring game and  the cut numbers, while the usual local value is related to the cut numbers of the bipartite graph constructed from the original graph.
This relationship makes it natural to define {\it quantum cut numbers of graphs} as follows.

\begin{defn} Given a graph $G=(V,E)$, a natural number $c \ge 2$ and for $t \in \{ q, qc \}$ we define the {\it t-quantum max c-cut number of G} to be
\[ Cut_{t,c}(G) = \frac{|E|}{2} \cdot \omega^s_t(Hom(G, K_c), \pi),\]
where $\pi$ is the uniform density on $E$.
\end{defn}

Using our characterizations of these synchronous values, we have that for a graph $G=(V,E)$ on $n$ vertices,
\begin{multline*} Cut_{qc,c}(G) =\frac{1}{2} \sup \{ \sum_{(x,y) \in E, a \ne b} \tau(e_{x,a}e_{y,b}) : \tau \in T(C^*(\bb F(n,c)) \} \\
= \frac{|E|}{2} - \frac{1}{2} \inf \{ \sum_{(x,y) \in E} \sum_{a=1}^c \tau(e_{x,a}e_{y,a}) : \tau \in T(C^*(\bb F(n,c)) \} \\
= \frac{1}{2}\inf \{ \| \sum_{(x,y) \in E, a \ne b} e_{e,a}e_{y,b} - C \| : C \in \cl C \}, \end{multline*}
while
\begin{multline*} Cut_{q,c}(G) = \frac{1}{2}\sup_n \{ \sum_{(x,y) \in E, a \ne b} tr_n(E_{x,a}E_{y,b}) : \{ E_{x,a} \} \text{ an (n,c)-PVM in } M_n \} 
\\= \frac{|E|}{2} - \frac{1}{2} \inf_n \{ \sum_{(x,y) \in E} \sum_{a=1}^c tr_n(E_{x,a}E_{y,a}) : \{ E_{x,a} \} \text{ an (n,c)-PVM in },\end{multline*}
where $tr_n$ denotes the normalized trace on $M_n$.

In a later section on XOR games we show that $Cut_{q,2}(G) = Cut_{qc,2}(G)$ and that this value is given by an SDP. There is a significant body of literature of semidefinite relaxations of max cut, for an introduction see \cite{La}. It is well-known that computing the classical max cut, $Cut_2(G)$, is an NP-hard problem.

\subsection{The Graph Correlation Function} This function, with a slightly different notation, was introduced and studied in \cite{DPP} where it was used to give a proof of the non-closure of $C_q^s(n,k)$ for all $n \ge 5, k \ge 2$. Given any graph $G=(V,E)$ and a C*-algebra with a trace $(\cl A, \tau)$ and a set of projections, $P_x \in \cl A, \, x \in V$, then the {\bf correlation} of these projections is
\[ \sum_{(x,y) \in E} \tau(P_xP_y).\]
Then for each $t \in \{ loc, q, qa, qc \}$ the {\bf graph correlation function $f_{G,t}(r)$} is defined as: 
\[ f_{G,t}(r)  = \inf \{ \sum_{(x,y) \in E} \tau(P_xP_y): \tau(P_x) =r, \forall x \in V \},\]
where the infimum is over all sets of projections $\{ P_x: x \in V \}$ in the C*-algebra and all traces of type t. Note that the C*-algebra is fixed and the optimization is over choices of projections and traces. So clearly,
\[0 \leq f_{G,qc}(r) \le f_{G, qa}(r) = f_{G,q}(r) \le f_{G,loc}(r),\] and there will exist projections and traces of type t attaining these values except, possibly,  in the case $q$.

In \cite{DPP}, it was shown that for the complete graph on 5 vertices, $K_5$, the value of the function $f_{K_5, q}(r)$ is not attained for any irrational value of $r$ in a certain interval, which was then shown to imply that $C_q(5,2)$ is not closed. 

In \cite{PSSTW} it was shown that if  we set
\[ r_{G,t} = \sup \{ r:  f_{G,t}(r) =0 \},\]
then
\[ r_{G,t}^{-1} \le \chi_t(G),\]
where these {\it quantum chromatic numbers} $\chi_t(G)$ of type $t \in \{loc,q,qa,qc\}$ is the least value of $c$ for which there exists a perfect strategy of type $t$ for the graph $c$-colouring game. In \cite{PSSTW} it is also shown that $r_{G,loc}^{-1}$ is equal to the fractional chromatic number of the graph $G$, while $r_{G,q}^{-1}$ agrees with the quantum fractional chromatic number introduced by D. Roberson\cite{Ro2013}.

In \cite{DPP} it is shown that if the infimum of the graph correlation function is attained by a set of projections $\{ P_x: x \in V \}$,  then for each $x \in V$,  $P_x$ commutes with $\sum_{y :(x,y) \in E} P_y$. In Section~7, we adapt their technique to obtain relations that must be satisfied by the projections that attain the synchronous value for other games.

We continue our study of the synchronous values of the $c$-colouring game by obtaining estimates in terms of the graph correlation function, which we will show are sharp for the case $c=2$.

 \subsection{The Uniform Synchronous Density}  The uniform distribution for $n$ inputs and $c$ outputs is given by  $p(a,b|x,y) = 1/c^2$, but this density is not synchronous.  We wish to introduce a synchronous anaolgue.

The {\bf uniform synchronous density on $n$ inputs and $c$ outputs} is given by the formula,
\[ p(a,b|x,y) = \begin{cases} 1/c^2, & x \ne y, \\ 1/c, & x=y, a=b, \\ 0, & x =y, a \ne b. \end{cases}\]

\begin{prop} The uniform synchronous density on $n$ inputs and $c$ outputs is a local density, i.e., is in $C^s_{loc}(n,c)$.
\end{prop}
\begin{proof}
Let $S = \{ (a_1,...,a_n): 0 \le a_i \le c-1, \, a_i \in \bb Z \}$ and define $S_{x,a} \subseteq S$ to be the $n$-tuples that are equal to $a$ in the $x$-th coordinate. Note that $\cup_{a=0}^{c-1} S_{x,a} =S$. Consider the uniform distribution $P$ on $S$ so that each point has probability $\frac{1}{|S|} = \frac{1}{c^n}$. 

On question pair $(x,y)$, Alice and Bob, using classical shared randomness, sample a tuple $(a_1,...,a_n)$ from $S$ according to $P$. Alice responds with $a_x$ and Bob responds with $a_y$.  This classical strategy generates the synchronous local density given by
\[ p(a,b|x,y) = \int_S \chi_{S_{x,a}} \chi_{S_{y,b}} dP = \frac{ |S_{x,a} \cap S_{y,b}|}{c^n},\]
where $\chi_T$ denotes the characteristic function of the set $T$.
It is easily checked that this is the uniform synchronous density.
\end{proof}

Somewhat surprisingly, another representation of the uniform synchronous density is given by the canonical trace on the free group $\bb F(n,c)$. Recall that the canonical trace on the algebra of a group $\bb C(G)$ is given by setting $\tau(u_e) = 1,$ where $e$ is the group identity, so that $u_e$ is the identity of $\bb C(G)$ and $\tau(u_g) =0, \forall g \ne e$, and extending linearly.  If $U_1,...,U_n$ are the order $c$ unitaries that generate $\bb F(n,c)$, then the canonical projections are given by
\[ e_{x,a} = \frac{1}{c} \sum_{j=0}^{c-1} \alpha^{-aj}U_x^j,\]
where $\alpha= e^{2 \pi i/c}$. Thus, $\tau(e_{x,a}) = 1/c$.
These projections and the canonical trace yield a synchronous density
\[ p(a,b|x,y) = \tau(e_{x,a}e_{y,b}),\]
which is easily seen to be the uniform synchronous density. It is somewhat remarkable that the trace arising from this free non-abelian group agrees on the generators, up to order two, with a trace arising from an abelian setting.  

This density gives us a bound on the graph correlation function.

\begin{prop} Let $G=(V,E)$be a graph on $n$ vertices. Then
\[ f_{G, loc}(1/c) \le \frac{|E|}{c^2}.\]
\end{prop}
\begin{proof} Let $E_{x,a}$ be the projections yielding the uniform synchronous density, then we have that
\[ f_{G, loc}(1/c) \le \sum_{(x,y) \in E} \tau(E_{x,1}E_{y,1}) = \frac{|E|}{c^2}.\]
\end{proof}

\begin{thm}\label{valuecorrelation} Let $G= (V,E)$ be a graph on $n$ vertices and consider the $c$-colouring game $Hom(G, K_c)$ played with the uniform distribution $\pi$ on $E$. Then for $t \in \{ loc, q, qc \}$, 
\begin{equation}\label{valuecorrelationinequality}max \{ 1 - \frac{1}{c}, \,\, 1- \frac{2}{|E|} f_{G,t}(1/2) \} \le \omega^s_t(Hom(G, K_c), \pi) \le 1 - \frac{c}{|E|} f_{G,t}(1/c).\end{equation}
\end{thm}
\begin{proof}
We see that
 the value of any synchronous density $p(a,b|x,y) \in C_t^s(n,k)$ is given by
\[ \omega(Hom(G,K_c), p) = 1 - \frac{1}{|E|} \sum_{a=0}^{c-1} \sum_{(x,y) \in E} p(a,a|x,y).\]
If we use the uniform synchronous density, then this becomes,
\[ 1 - \frac{1}{|E|} \sum_{a=0}^{c-1} \sum_{(x,y) \in E} 1/c^2 = 1 - \frac{1}{c}.\]

If we assume that our density is synchronous so that there exist PVM's $\{ E_{x,a}: 0 \le a \le c-1 \}$ such that $p(a,b|x,y) = \tau(E_{x,a}E_{y,b})$ for some C*-algebra and trace $\tau: \cl A \to \bb C$ of type t, then we have that
\begin{multline}
\sum_{a=0}^{c-1} \sum_{(x,y) \in E} p(a,a|x,y) = \sum_{(x,y) \in E} \sum_{a=0}^{c-1} \tau(E_{x,a}E_{y,a})= c \sum_{(x,y) \in E} \tau^{(c)} (P_xP_y), \end{multline}
where we set $\cl A^{(c)} = \cl A \oplus \cdots \oplus \cl A$(c times) and let $\tau^{(c)}: \cl A^{(c)} \to \bb C$ be the unital trace $\tau^{(c)}(X_0 \oplus \cdots \oplus X_{c-1}) = 1/c \sum_{a=0}^{c-1} \tau(X_a)$ and let $P_x= E_{x,0} \oplus \cdots \oplus E_{x,c-1}$.  Note that in this case, for every $x$, we have that
\[ \tau^{(c)}(P_x) = 1/c \sum_{a=0}^{c-1} \tau(E_{x,a}) = 1/c.\]
This proves that 
\[ \omega^s_t(Hom(G, K_c)) \le 1- \frac{c}{|E|} f_{G,t}(1/c).\]

For the other inequality, suppose that we are given projections, $\{ P_x: x \in V \} \subseteq \cl A$ and a trace $\tau$ of type t with $\tau(P_x) = 1/2$.  Then we set $E_{x,0} = P_x, \, E_{x,1} = I- P_x$ and $E_{x,a} =0, \, a \ne 0,1$. For the corresponding synchronous correlation, we have that
\begin{align*} 1- \omega^s_t(Hom(G,K_c)) &\le \frac{1}{|E|} \sum_{(x,y) \in E} \sum_a \tau(E_{x,a}E_{y,a}) \\&= 
\frac{1}{|E|} \sum_{(x,y) \in E} \tau( P_xP_y + (I-P_x)(I-P_y)) \\&= \frac{1}{|E|} \sum_{(x,y) \in E} \tau(2P_xP_y + I - P_x - P_y)\\ &=\frac{2}{|E|} \sum_{(x,y) \in E} \tau(P_xP_y),
\end{align*}
and the other inequality follows.
\end{proof}

\begin{cor} Let $G$ be a graph on $n$ vertices. Then for the 2-colouring game, with uniform distribution on E, we have that
\[ \omega^s_t(Hom(G, K_2)) = 1 - \frac{2}{|E|} f_{G,t}(1/2).\]
In particular, $\omega^s_q(Hom(G,K_2)) = \omega^s_{qc}(Hom(G, K_2))$ and $Cut_{q,2}(G) = \frac{|E|}{2}- f_{G,q}(1/2)$.
\end{cor}
\begin{proof} The first result follows from the above inequalities. The second follows from \cite[Proposition~3.10]{DPP} where it is shown that for any graph, $f_{G,q}(1/2) = f_{G, qc}(1/2)$.
\end{proof}

There are similar inequalities, with different constants, for each of the three types of densities discussed at the beginning of the section.

In general for $c \ne 2$, we do not expect that the upper bound is sharp. For example, suppose that we had a graph such that
\[ r_{G,t}^{-1} \le c < \chi_t(G).\] Then $f_{G,t}(1/c) =0$, but since $c < \chi_t(G)$ there is no perfect t-strategy and hence,
\[ \omega^s_t(Hom(G, K_c)) < 1 = 1 - \frac{c}{n^2} f_{G,t}(1/c).\]
Unfortunately, we do not know an example of a graph with this particular separation, so we cannot say definitely that $\omega^s_t(Hom(G, K_c)) \ne 1 - \frac{c}{n^2 f_{G,t}(1/c)},$ for some $c$.

It is a consequence of Tsirelson's work that for any graph $f_{G,q}(1/2) = f_{G,qc}(1/2)$, this is mentioned in \cite{DPP} and we provide another proof in Section \ref{sec:syncxor}. Consequently, for the uniform distribution,
\[ \omega^s_q(Hom(G, K_2)) = \omega^s_{qc}(Hom(G,K_2)).\]
In fact, Tsirelson's work tells us quite a bit more in the 2-colouring case, since 2-colouring games, with appropriately chosen distributions on questions, belong to a family of games known as XOR games, which is the topic of our next section. 

First, we consider  the value of the game of $c$-colouring a complete graph on $n$ vertices when $n > c$.

\subsection{c-Colouring the Complete Graph on n Vertices}

We now turn our attention to the case that $G= K_n$. We begin by computing the graph correlation function in this case. In addition to the graph correlation functions, $f_{G,t}(r), t = loc, q, qc$,
the paper \cite{DPP} also introduces a function $f_{G,vect}(r)$ that satisfies, $f_{G,vect}(r) \le f_{G,qc}(r)$. We use this fact in the proof of the following theorem.

\begin{thm}\label{graphcorrelation} For the complete graph $K_n, \, n \ge 5$ and $\frac{n - \sqrt{n^2-4n}}{2n} \le r \le \frac{n + \sqrt{n^2-4n}}{2n}$ we have that
\[ f_{K_n,q}(r) =f_{K_n,qc}(r) = nr(nr -1).\]
\end{thm}
\begin{proof} 
In \cite[Proposition 4.1]{DPP} it is shown that for the complete graph $f_{K_n, vect}(r) = nr(nr-1)$ for $\frac{1}{n} \le r \le \frac{n-1}{n}.$  Note that $\frac{1}{n} \le \frac{n - \sqrt{n^2 - 4n}}{2n}$ and $\frac{n + \sqrt{n^2 - 4n}}{2n} \le \frac{n-1}{n}$. 

In  \cite{KRS}, it is proven that for any rational $r$ in this smaller interval there exist $n$ projection matrices  in $M_m$ for some $m$, $Q_x, 0 \le x \le n-1$ such that $\sum_{x=0}^{n-1} Q_x = (nr)I_m$. Let 
\[P_x = \oplus_{j=0}^{n-1} Q_{j+x},\]
where the index is modulo $n$.  Then $\sum_{x=0}^{n-1} P_x = (nr)I_{nm}$.  Moreover, if we let $\tau$ denote the normalized trace on $M_{mn}$ then $\tau(P_x) = r$ for every $x$. 
Thus we can write
\[
nr(nr-1) =f_{K_n, vect}(r) \le f_{K_n,qc}(r) \le f_{K_n,q}(r) \le \sum_{(x,y) \in E} \tau(P_xP_y).
\]
Now notice that
\begin{align*}
\sum_{(x,y) \in E} \tau(P_xP_y) &= \sum_{x=0}^{n-1} \sum_{y \ne x} \tau(P_xP_y)\\ &= \sum_{x=0}^{n-1} \tau(P_x( (nr)I_{nm} - P_x))\\ &= \sum_{x=0}^{n-1} (nr -1) \tau(P_x) \\&= nr(nr -1). 
\end{align*}
The result follows by observing that  the functions $f_q= f_{qa}$ and $f_{qc}$ are continuous.
\end{proof}

In order for $r=1/c$ to satisfy the inequality of Theorem~\ref{graphcorrelation}, it is necessary and sufficient that $\frac{c^2}{c-1} \le n$. This is satisfied if $c \le n-2$.

\begin{thm} Let $c \le n-2$ and $n \ge 5$. Then
\begin{multline*}
 \omega^s_{loc}(Hom(K_n, K_c)) =
\omega^s_q(Hom(K_n,K_c)) = \omega^s_{qc}(Hom(K_n,K_c)) = 1 + \frac{1}{n} - \frac{1}{c}.\end{multline*}
\end{thm}
\begin{proof}
By Theorem~\ref{valuecorrelation}, we have that 
\[ 1- \frac{1}{c} + \frac{1}{n} = 1 - \frac{n^2-n}{n^2c} = 1 - \frac{|E|}{n^2c} \le \omega^s_{loc}(Hom(K_n,K_c)).\] 

On the other hand,  by Theorem~\ref{valuecorrelation} and Theorem~\ref{graphcorrelation},
\[ \omega^s_{qc}(Hom(K_n,K_c)) \le 1- \frac{c}{n^2} f_{K_n, qc}(1/c) = 1 - \frac{c}{n^2} (n/c(n/c -1) = 1- \frac{1}{c} + \frac{1}{n}.\]
\end{proof}

\section{Synchronous Values of XOR Games}\label{sec:syncxor}
In \cite{CSUU} quantum values of XOR games were studied extensively.
In this section, we recall their results, study synchronous values of XOR games, explain how to calculate the synchronous values using semidefinite programming, and compare the two sets of results.  Later, we will consider several specific examples of synchronous values of XOR games and study their properties. For XOR games the output set is always $\bb Z_2$.

\begin{defn}
A game $G=(I,\{0,1\},\lambda)$ is an \textbf{XOR game} if there exists a function $f: I \times I \to \{0,1\}$ such that $\lambda(x,y,a,b)=1$ if and only if $a \oplus b = f(x,y)$, where $a \oplus b$ denotes addition in the binary field.
\end{defn}  

Note that an XOR game is synchronous if and only if $f(x,x) =0$ for all $x \in I$, and symmetric if and only if $f(x,y) = f(y,x)$.

Computing values of XOR games is especially straightforward, because of the following observation together with the Tsirelson's theory. 

\begin{prop} \label{prop: unbiased strategies for XOR}
Let $G$ be an XOR game with $|I|=n$ and prior distribution $\pi$, and let $t \in \{loc, qa, qc\}$. Then there exists a strategy $p \in C_t(n,2)$ such that $\omega_t(G,\pi) = \omega(G, \pi, p)$, where $p_A(0|x) = p_B(0|y) = 1/2$ for each $x,y \in I$.
\end{prop}

\begin{proof}
Since $C_t(n,2)$ is closed for each $t \in \{loc, qa, qc\}$, there exists $p \in C_t(n,2)$ such that $\omega_t(G,\pi) = \omega(G, \pi, p)$. Given such a density $p$, there exists a Hilbert space $H$, operators $P_1, \dots, P_n, Q_1, \dots, Q_n \in B(H)$, and a unit vector $h \in H$ such that 
\[ p(0,0|x,y) = \langle P_x Q_y h, h \rangle \]
for each $x,y \in I$. For each $x \in I$, define $P_i' = P_i \oplus (I - P_i)$ and $h' = \frac{1}{\sqrt{2}}(h \oplus h)$. Let $p' \in C_t(n,2)$ be the unique density satisfying
\[ p'(0,0|x,y) = \langle P_x' Q_y' h', h' \rangle \]
for each $x,y \in I$. Note that $p'(a,b|x,y) = \frac{1}{2}(p(a,b|x,y) + p(a \oplus 1, b \oplus 1|x,y))$. Then
\begin{eqnarray} 
\omega(G, \pi, p) & = & \sum_{x,y \in I, a,b \in \{0,1\}} \pi(x,y) p(a,b|x,y) \lambda(x,y,a,b) \nonumber \\
& = & \sum_{x,y \in I, a \oplus b = f(x,y)} \pi(x,y) p(a,b|x,y) \nonumber \\
& = & \sum_{x,y \in I, a \oplus b = f(x,y)} \pi(x,y) \frac{1}{2}(p(a,b|x,y) + p(a \oplus 1, b \oplus 1|x,y)) \nonumber \\
& = & \sum_{x,y \in I, a \oplus b = f(x,y)} \pi(x,y) p'(a,b|x,y) \nonumber \\
& = & \omega(G, \pi, p') \nonumber
\end{eqnarray}
where we have used the fact that $a \oplus b = (a \oplus 1) \oplus (b \oplus 1)$. Since $p'_A(0|x) = p'_B(0|y) = 1/2$ and since $\omega_t(G,\pi) = \omega_t(G, \pi, p) = \omega_t(G,\pi,p')$, the statement is proven.
\end{proof}

Two-outcome densities satisfying $p_A(0|x)=p_B(0|y) = 1/2$ for all $x,y \in I$ are called \textbf{unbiased} densities in the literature. The following theorem is a restatement of Tsirelson's characterisation of quantum observables \cite{Ts} in terms of unbiased densities. For those unfamiliar with the similarities and differences between quantum observables and quantum densities see \cite[Theorem 11.8]{Pa-ENL}.

\begin{thm}[Tsirelson] \label{thm: Tsirelsons trick}
Let $p(i,j|s,t)$ be a density such that $p_A(0|s) = p_B(0|t) = 1/2$ for all $s,t$.  Then the following statements are equivalent:
\begin{enumerate}
\item $p(i,j|s,t) \in C_{qc}(n,2)$.
\item
There exist real unit vectors $x_s,y_t$ for $1 \leq s,t \leq n$ such that $p(i,j|s,t)=\frac{1}{4}[1+(-1)^{i+j}\langle x_s,y_t \rangle]$.
\item
$p(i,j|s,t) \in C_q(n,2)$.
\end{enumerate}
\end{thm}

A similar statement can be made in the synchronous case.

\begin{thm} \label{thm: synchronous Tsirelsons trick}
Let $p(i,j|s,t)$ be a synchronous density such that $p(0,0|s,s)=p(1,1|s,s)$ for all $s$.  Then the following statements are equivalent:
\begin{enumerate}
\item $p(i,j|s,t) \in C_{qc}^s(n,2)$.
\item
There exist real unit vectors $x_s$ for $1 \leq s \leq n$ such that  $p(i,j|s,t)=\frac{1}{4}[1+(-1)^{i+j}\langle x_s,x_t \rangle]$.
\item
$p(i,j|s,t) \in C_q^s(n,2)$.
\end{enumerate}
\end{thm}

\begin{proof}
Suppose the first statement is true. By Theorem \ref{thm: Tsirelsons trick}, there exist unit vectors $x_s,y_t$ for $1 \leq s,t \leq n$ such that $p(i,j|s,t)=\frac{1}{4}[1+(-1)^{i+j}\langle x_s,y_t \rangle]$. Since $p(i,j|s,s)=0$ whenever $i \neq j$, we have $\langle x_s, y_s \rangle = 1$ for every $s$. By Cauchy-Schwarz, $x_s = y_s$ for every $s$. The other implications are straightforward.
\end{proof}

\begin{remark} \label{graphcorrathalf} Given projections $P_x$ in a C*-algebra with a trace $(\cl A, \tau)$ such that $\tau(P_x) = 1/2$, set $E_{x,0} = P_x$ and $E_{x,1} = I - P_x$. Then $\tau(E_{x,i}E_{y,j}):=p(i,j|x,y)$ is a density in $C_{qc}$ with marginals equal to $1/2$. Hence by the above result $p(i,j|x,y) \in C_q$. Give a graph $G=(V,E)$, to compute $f_{G,qc}(1/2)$ we are minimizing
\[ \sum_{(x,y) \in E} \tau(P_xP_y) = \sum_{(x,y) \in E} p(0,0|x,y),\]
over all sets of projections with $\tau(P_x) = 1/2$
and, hence, $f_{G,qc}(1/2) = f_{G,q}(1/2)$. This is essentially the proof given in \cite[Proposition~3.10]{DPP}.
\end{remark}

We will use the theorems above, together with Proposition \ref{prop: unbiased strategies for XOR}, to calculate the values of certain XOR games. For now, we will only provide a general formulation for these values in terms of semidefinite programs.

\begin{remark}
Let $G=(I,\{0,1\},\lambda)$ be an XOR game with $n := |I|$, and suppose $f:I \times I \to \{0,1\}$ is a function satisfying $f(x,y)=a \oplus b$ if and only if $\lambda(x,y,a,b)=1$ for all $a,b \in \{0,1\}$ and $x,y \in I$. Let $\pi(x,y)$ be a prior distribution on $I$, and let $\cl G= (G,\pi)$ denote the game $G$ with questions asked according to the distribution $\pi$. Following \cite{CSUU}, we define the matrix $A_{\cl G} \in M_n$ by $A_{\cl G} = ((-1)^{f(x,y)}\pi(x,y))$, which \cite{CSUU} call the {\bf cost matrix}. They also study a matrix 
\[ B_{\cl G} := \frac{1}{2} \begin{pmatrix} 0 & A_{\cl G} \\ A_{\cl G}^T & 0 \end{pmatrix} \in M_{2n}. \]

For synchronous values, the matrix,
 \[A_{\cl G}^s := \frac{1}{2}(A_{\cl G} + A_{\cl G}^T) \in M_n\]
 plays a similar role to the cost matrix and we will refer to this matrix as the {\bf symmetrized cost matrix}. 

\end{remark}

Let $\mathcal{E}_n \subseteq M_n$ denote the $n \times n$ \textbf{elliptope} defined by 
\begin{equation} \label{eq: elliptope}
\mathcal{E}_n := \{ P \in M_n(\mathbb{R}) : diag(P) = I_n \text{ and } P\geq 0 \}.
\end{equation}

The following formula for the value of an XOR game is a restatement of results in \cite{CSUU}. The formula for the synchronous value is new.

\begin{thm} \label{thm: semidefinite program for XOR games}
Let $G=(I,\{0,1\},\lambda)$ be an XOR game with $n := |I|$. Let $\pi(x,y)$ be a prior distribution on $I$. Then
\[ \omega_{qc}(G,\pi) = \omega_{q}(G,\pi) = \frac{1}{2} + \frac{1}{2} \max_{P \in \mathcal{E}_{2n}} Tr(B_{\cl G} P) \]
and
\[ \omega_{qc}^s(G,\pi) = \omega_{q}^s(G,\pi) = \frac{1}{2} +\frac{1}{2} \max_{P \in \mathcal{E}_{n}} Tr(A_{\cl G}^s P). \]
\end{thm}
\begin{proof}
Suppose $f:I \times I \to \{0,1\}$ is a function satisfying $f(x,y)=a \oplus b$ if and only if $\lambda(x,y,a,b)=1$ for all $a,b \in \{0,1\}$.

We first consider the claim concerning $\omega_{qc}(G,\pi)$. By Proposition \ref{prop: unbiased strategies for XOR}, there exists $p \in C_q(n,2)$ such that $\omega_{qc}(G,\pi) = \omega(G,\pi,p)$ and $p_A(0|x)=p_B(0|y)=1/2$ for every $x,y \in I$. Since $\lambda(x,y,a,b) = 1$ if and only if $a \oplus b = f(x,y)$, we have that
\[ \omega_{qc}(G,\pi) = \sum_{x,y \in I, a\oplus b= f(x,y)} \pi(x,y) p(a,b|x,y). \]
By Theorem \ref{thm: Tsirelsons trick} this implies
\begin{eqnarray}
\omega_{qc}(G,\pi) & = & \sum_{x,y \in I, a\oplus b= f(x,y)} \frac{1}{4} \pi(x,y)(1 + (-1)^{a+b} \langle v_x, w_y \rangle) \nonumber \\
& = & \frac{1}{4}\sum_{x,y \in I, a\oplus b= f(x,y)} \pi(x,y) + \frac{1}{4} \sum_{x,y \in I} \pi(x,y) (-1)^{f(x,y)} \langle v_x, w_y \rangle \nonumber
\end{eqnarray}
where the $v_x$'s and $w_y$'s are real unit vectors. Since every expression of the form $p(a,b|x,y) = \frac{1}{4}[1+(-1)^{a+b}\langle v_x,w_y \rangle]$ defines an element of $C_{qc}(n,2)$, we have
\[ \omega_{qc}(G,\pi) = \frac{1}{4}\sum_{x,y \in I, a\oplus b= f(x,y)} \pi(x,y) + \frac{1}{4} \max_{v_x, w_y} \sum_{x,y \in I} \pi(x,y) (-1)^{f(x,y)} \langle v_x, w_y \rangle \]
where the maximization is over all sets of real unit vectors $v_x$ and $w_y$. Since $\pi(x,y)$ is a probability distribution and $a \oplus b = f(x,y)$ for exactly two choices of pairs $(a,b)$, we have that
\[ \sum_{x,y \in I, a\oplus b= f(x,y)} \pi(x,y) = 2. \]
Also, notice that an $n \times n$ matrix has the form $(\langle v_x, w_y \rangle )_{x,y}$ for unit vectors $v_x$ and $w_y$ if and only if it is the upper right (or lower left) $n \times n$ corner of a matrix $P \in \mathcal{E}_{2n}$, since every element $P \in \mathcal{E}_{2n}$ has a Gram decomposition 
\[ P = (v_1 \dots v_n w_1 \dots w_n)^*(v_1 \dots v_n w_1 \dots w_n). \]
A computation yields the expression
\[ \omega_{qc}(G,\pi) = \omega_{q}(G,\pi) = \frac{1}{2} + \frac{1}{2} \max_{P \in \mathcal{E}_{2n}} Tr(B_{\cl G} P). \]

To verify the claims concerning $\omega_{qc}^s(G,\pi)$, first note that by the above argument we have
\[ \omega_{qc}^s(G,\pi) = \omega_{q}^s(G,\pi) = \frac{1}{2} + \frac{1}{2} \max_{P' \in \mathcal{E}_{2n}'} Tr(B_{\cl G} P'). \]
where $\mathcal{E}_{2n}' \subseteq \mathcal{E}_{2n}$ is taken to be the set of $P \in \mathcal{E}_{2n}$ whose upper right $n \times n$ corner has the form $(\langle v_x, v_y \rangle )_{x,y}$ for a single set of real unit vectors $\{v_1, \dots, v_n\}$, by Theorem \ref{thm: synchronous Tsirelsons trick}. Because of the form of $B_{\cl G}$, we may assume any $P' \in \mathcal{E}_{2n}'$ has the form
\[ P' = \begin{pmatrix} P & P \\ P & P \end{pmatrix}, \quad P \in \mathcal{E}_n, \]
and a computation shows that $Tr(B_{\cl G} P') = Tr(A_G^s P)$. Thus
\[ \omega_{qc}^s(G,\pi) = \omega_{q}^s(G,\pi) = \frac{1}{2} + \frac{1}{2} \max_{P \in \mathcal{E}_{n}} Tr(A_{\cl G}^s P). \]
This proves the claims.
\end{proof}


\section{Two Colourings} 

The 2-colouring game for a graph $G=(V,E)$ is not formally an XOR game, since whenever $x \ne y$ and $(x,y) \notin E$ we have that $\lambda(x,y,a,b) = 1$ for all pairs $a,b$, while an XOR game requires that $a \oplus b = f(x,y) \in \{0,1\}$ to win, for every $x,y \in V$. However, if the prior distribution on inputs has the property that $\pi(x,y) =0,$ whenever $x \ne y$ and $(x,y) \notin E$, then we may arbitrarily set $f(x,y)$ to be 0 or 1, without altering the corresponding value of the game. Thus, when we restrict to prior distributions with this property, we may apply the results on synchronous XOR games to compute the value of 2-colouring games.

\begin{prop} Let $G=(V,E)$ be a graph on $n$ vertices and let $A_G$ denote its adjacency matrix.  Then
\[ Cut_{q,2}(G) = Cut_{qc,2}(G) = \frac{|E|}{4} - \frac{1}{4} \min_{P \in \cl E_n} Tr(A_G P).  \]
\end{prop}
\begin{proof} Recall that to compute this value we consider the game $\cl G= (Hom(G,K_2), \pi)$ where $\pi$ is the uniform density on $E$.  In this case we have an XOR game with $f(x,y)=1, \, \forall (x,y) \in E$ and 0 otherwise.  Thus,
$A_{\cl G}^s= \big( (-1)^{f(x,y)} \pi(x,y) \big) = \frac{-1}{|E|} A_G$ and the result follows by Theorem~\ref{thm: semidefinite program for XOR games}.
\end{proof}

It is not hard to see that if we let $\cl P_n \subset \cl E_n$ be the set of all rank one positives all of whose entries are $\pm 1$, then the ordinary max cut is given by
\[ Cut_2(G) = \frac{|E|}{4} - \frac{1}{4} \min_{P \in \cl P_n} Tr(A_GP).\]
  This gives another way to see $Cut_{q,2}(G)$ as a relaxation of the usual max cut.

We now turn our attention to studying 2 colourings for odd cycles.
Let $C_{2k+1}$ be an odd cycle. We will index the vertices by $\bb Z_{2k+1}$ so that vertices are adjacent if and only if they are the pair $(j, j \pm 1), \, 0 \le j \le 2k$ where $2k+1 =0$. We consider the game $\cl G= Hom(C_{2k+1}, K_2)$ with several different prior distributions on $\mathbb{Z}_{2k+1} \times \mathbb{Z}_{2k+1}$. We first consider a non-symmetric uniform distribution, first studied by Cleve-Hoyer-Toner-Watrous \cite{CHTW}, in order to compare the synchronous and non-synchronous values of the game. We then consider a natural family of symmetric distributions. We will show that for both non-symmetric and symmetric distributions, the synchronous quantum value of the game can be strictly smaller than the quantum value of the game, though in some cases these values may coincide. In all cases, the $q$ and $qc$ values of the game will coincide.

\subsection{Non-symmetric uniform distribution}

We now compute the synchronous $q$-value of $\cl G$ with the prior distribution given by
\[ \pi(x,y) = \begin{cases} \frac{1}{2n} & x = y \text{ or } x + 1 = y \mod{n} \\ 0 & \text{else} \end{cases}\]
where $n=2k+1$. The game $\cl G$ with this distribution was studied in Subsection 3.2 of \cite{CHTW}, where it was show that
\[ \omega_{qc}(\cl G) = \omega_q(\cl G) = \cos^2(\pi/4n) = \frac{1}{2} + \frac{1}{2}\cos(\pi/2n). \]
We will show that $\omega_{qc}^s(\cl G)=\omega_{q}^s(\cl G) = \frac{1}{2} + \frac{1}{2}\cos^2(\pi/2n)$, which is strictly less than $\omega_{qc}(\cl G)$.

\begin{thm} \label{thm: sync value non-symmetric pi}
Let $n = 2k+1$. Then $\omega_{qc}^s(\cl G)=\omega_{q}^s(\cl G) = \frac{1}{2} + \frac{1}{2}\cos^2(\pi/2n)$.
\end{thm}

\begin{proof}
By Theorem \ref{thm: semidefinite program for XOR games}, we have
\[ \omega_{qc}^s(\cl G, \pi) = \omega_{q}^s(\cl G, \pi) = \frac{1}{2} + \frac{1}{2} \max_{P \in \mathcal{E}_n} Tr(A_{\cl G}^s P) \]
where 
\[ A_{\cl G}^s = \begin{pmatrix} \frac{1}{2n} & -\frac{1}{4n} & 0 & \dots & -\frac{1}{4n} \\
-\frac{1}{4n} & \frac{1}{2n} & -\frac{1}{4n} & \dots & 0 \\
 & \ddots & \ddots & \ddots & \\
 0 &  & -\frac{1}{4n} & \frac{1}{2n} & -\frac{1}{4n} \\
 -\frac{1}{4n} & \dots & 0 & -\frac{1}{4n} & \frac{1}{2n} 
\end{pmatrix} \]
and $\mathcal{E}_n$ denotes the $n \times n$ elliptope defined in Equation (\ref{eq: elliptope}). Thus, it suffices to calculate \[ \max_{P \in \mathcal{E}_n} Tr(A_{\cl G}^s P). \] The value of this semidefinite program is equal to the value of the dual program
\[ \min_{D \in \mathcal{D}_n} Tr(D) \quad \text{subject to} \quad D - A_{\cl G}^s \geq 0 \]
where $\mathcal{D}_n$ denotes the set of $n \times n$ diagonal real matrices. By the symmetry of $A_{\cl G}^s$, it suffices to minimize $Tr(D)$ over all constant diagonal matrices. This is because if $D$ is diagonal and satisfies $D - A_{\cl G}^s \geq 0$, then $U^*(D-A_{\cl G}^s)U = U^*DU - A_{\cl G}^s \geq 0$ where $U$ is the cyclic shift
\[ U = \begin{pmatrix} 0 & 1 & 0 & \dots & 0 \\
0 & 0 & 1 & \dots & 0 \\
 & \ddots & \ddots & \ddots & \\
 0 &  & 0 & 0 & 1 \\
 1 & 0 & \dots & 0 & 0 
\end{pmatrix} \] 
Averaging $(U^j)^*D(U^j)$ over a $j \in \{0,1,\dots, n-1\}$ yields a constant matrix with the same trace as $D$. Hence, we only need to calculate
\[ \min_{y \in \mathbb{R}} ny \quad \text{subject to} \quad yI_n - A_{\cl G}^s \geq 0. \] Since the matrix $yI_n - A_{\cl G}^s$ is circulant, its eigenvalues have the form \[ \lambda_j = (y - \frac{1}{2n}) + \frac{1}{4n} \omega_n^j + \frac{1}{4n} \omega_n^{(n-1)j}, \] where $\omega_n = e^{2 \pi i / n}$ is the primitive $n$-th root of unity (c.f. Exercise 2.2P10 of \cite{HJbook}). Observe that $\lambda_j$ is real since $\omega_n^{-j} = \omega_n^{(n-1)j}$ and thus $\omega_n^j + \omega_n^{(n-1)j} = 2 \text{Re}(\omega_n^j)$. The smallest value of $y$ for which $\lambda_j \geq 0$ for every $j$ is
\[ y = \frac{1}{2n} + \frac{1}{2n}\cos(\pi/n). \]
It follows that
\[ \max_{P \in \mathcal{E}_n} Tr(A_{\cl G}^s P) = \frac{1}{2}(1 + \cos(\pi/n)). \]
Consequently,
\begin{eqnarray}
\omega_{qc}^s(G, \pi) & = & \frac{1}{2} + \frac{1}{4} \left[ 1 + \cos(\pi/n) \right] \nonumber \\
& = & \frac{1}{2} + \frac{1}{4} \left[ 1 + 2\cos^2(\pi/2n) - 1 \right] \nonumber \\
& = & \frac{1}{2} + \frac{1}{2} \cos^2(\pi/2n) \nonumber
\end{eqnarray}
as desired.
\end{proof}

\subsection{Symmetric distributions}

The above shows that the synchronous $q$-value of a game is sometimes strictly smaller than the $q$-value of the game. In that case, the gap between these values is aided by the fact that the prior distribution is not symmetric. We will now show that even when the prior distribution is symmetric, there may still be a gap between the synchronous $q$-value of the game and the $q$-value of the game.

Let $p,q \geq 0$ with $p+q = 1$. Consider the symmetric prior distribution
\begin{equation} \label{eqn: symmetric prior distribution}
\pi(x,y) = \begin{cases} \frac{p}{2n} & x+1 = y \mod n \\ \frac{p}{2n} & y+1 = x \mod n \\ \frac{q}{n} & x=y \\ 0 & \text{else} \end{cases} \end{equation}
where $n=2k+1$. We first calculate the $q$-value of the two-colouring game, which is again equal to the $qc$-value of the game. 

\begin{thm} \label{thm: non-sync value with symmetric distribution}
Let $p,q \geq 0$ with $p+q = 1$, and let $\pi$ be the prior distribution given in equation (\ref{eqn: symmetric prior distribution}),
where $n=2k+1$. Then
\[ \omega_{qc}(\cl G) = \omega_q(\cl G) = \begin{cases} p & p > \frac{1}{2-\cos^2(\pi/2n)} \\ q + p\cos^2(\pi/2n) & \text{else}. \end{cases} \]
Moreover, $\omega_{qc}(\cl G) = \omega_{loc}(\cl G)$ whenever $p > \frac{1}{2-\cos^2(\pi/2n)}$.
\end{thm}

\begin{proof}
By Theorem \ref{thm: semidefinite program for XOR games}, we have
\[ \omega_{qc}(\cl G, \pi) = \omega_{q}(\cl G, \pi) = \frac{1}{2} + \frac{1}{2} \max_{P \in \mathcal{E}_{2n}} Tr(B_{\cl G} P) \]
where 
\[ B_{\cl G} := \frac{1}{2} \begin{pmatrix} 0 & A_{\cl G} \\ A_{\cl G}^T & 0 \end{pmatrix} \in M_{2n}, \quad A_{\cl G} = \begin{pmatrix} \frac{q}{n} & -\frac{p}{2n} & 0 & \dots & -\frac{p}{2n} \\
-\frac{p}{2n} & \frac{q}{n} & -\frac{p}{2n} & \dots & 0 \\
 & \ddots & \ddots & \ddots & \\
 0 &  & -\frac{p}{2n} & \frac{q}{n} & -\frac{p}{2n} \\
 -\frac{p}{2n} & \dots & 0 & -\frac{p}{2n} & \frac{q}{n} 
\end{pmatrix} \]
and $\mathcal{E}_{2n}$ denotes the $2n \times 2n$ elliptope. We will now calculate
\[ \max_{P \in \mathcal{E}_{2n}} Tr(B_{\cl G} P). \]
The value of this semidefinite program is equal to the value of the dual program
\[ \min_{D \in \mathcal{D}_{2n}} Tr(D) \quad \text{subject to} \quad D - B_{\cl G} \geq 0 \]
where $\mathcal{D}_{2n}$ denotes the set of $2n \times 2n$ diagonal real matrices. By the symmetry of $B_{\cl G}$, it suffices to minimize $Tr(D)$ over all constant diagonal matrices. Hence, we only need to calculate
\[ \min_{y \in \mathbb{R}} 2ny \quad \text{subject to} \quad yI_{2n} - B_{\cl G} \geq 0. \] It follows from Lemma 3.1 of \cite{Pa-CBM} that the value of this semidefinite program is 
\[ 2n \| B_{\cl G} \| = n \|A_{\cl G}\|. \]
Since $A_{\cl G}$ is symmetric, its norm is equal to $\max_j |\lambda_j|$, where $\lambda_0, \lambda_1, \dots, \lambda_{n-1}$ are the eigenvalues of $A_{\cl G}$. Since $A_{\cl G}$ is circulant, its eigenvalues have the form
\[ \lambda_j = \frac{q}{n} - \frac{p}{2n} \omega_n^j - \frac{p}{2n} \omega_n^{(n-1)j} \]
where $\omega_n = e^{2 \pi i/n}$ is the $n$-th root of unity. Thus, the smallest eigenvalue of $A_{\cl G}$ is $\lambda_0 = \frac{q-p}{n}$, while the largest eigenvalue is $\lambda_{(n-1)/2} = \frac{q}{n} + \frac{p}{n} \cos(\pi/n)$. A calculation shows that
\[ \frac{p-q}{n} > \frac{q}{n} + \frac{p}{n} \cos(\pi/n) \quad \text{if and only if} \quad p > \frac{2}{3-\cos(\pi/n)} = \frac{1}{2-\cos^2(\pi/2n)} \]
using $q = 1 - p$. Thus
\[ n\|A_{\cl G}\| = \begin{cases} p-q & p > \frac{1}{2-\cos^2(\pi/2n)} \\ q + p \cos(\pi/n) & \text{else} \end{cases} \]
and thus
\[ \omega(\cl G, \pi) = \begin{cases} \frac{1}{2} + \frac{1}{2} \left[ p-q \right] & p > \frac{1}{2-\cos^2(\pi/2n)} \\ \frac{1}{2} + \frac{1}{2} \left[ q + p \cos(\pi/n) \right] & \text{else} \end{cases}. \]
Since
\[ \frac{1}{2} + \frac{1}{2} \left[ p-q \right] = \frac{1}{2}(p+q) + \frac{1}{2}(p-q) = p \]
and
\begin{eqnarray}
\frac{1}{2} + \frac{1}{2} \left[ q + p \cos(\pi/n) \right] & = & \frac{1}{2} + \frac{1}{2}(1-p) + \frac{p}{2}(\cos(\pi/n)) \nonumber \\
& = & 1 - \frac{p}{2} + \frac{p}{2}(2\cos^2(\pi/2n)-1) \nonumber \\
& = & 1-p + p \cos^2(\pi/2n) \nonumber \\
& = & q + p \cos^2(\pi/2n), \nonumber
\end{eqnarray}
the first statement is proven. That $\omega_{loc}(\cl G, \pi) = p$ when $p > \frac{1}{2-\cos^2(\pi/2n)}$ follows from the observation that the value $p$ is obtained when Alice and Bob employ the deterministic strategy of always returning opposite colors.
\end{proof}

We remark that whenever $p > (2-\cos^2(\pi/2n))^{-1}$, the winning deterministic strategy of always returning the opposite color is not a synchronous strategy. Let us now consider the synchronous value of this game.

\begin{thm}
Let $p,q \geq 0$ with $p+q = 1$, and let $\pi$ be the prior distribution given in equation (\ref{eqn: symmetric prior distribution}),
where $n=2k+1$. Then
\[ \omega_{qc}^s(\cl G) = \omega_q^s(\cl G) = q + p\cos^2(\pi/2n). \]
Consequently, $\omega_{qc}^s(\cl G) < \omega_{qc}(\cl G)=\omega_{loc}(\cl G)$ whenever $p > \frac{1}{2-\cos^2(\pi/2n)}$.
\end{thm}

\begin{proof}
The proof is similar to the proof of Theorem \ref{thm: sync value non-symmetric pi}, so we just outline the main points. By Theorem \ref{thm: semidefinite program for XOR games},
\[ \omega_{qc}^s(\cl G, \pi) = \omega_{q}^s(\cl G, \pi) = \frac{1}{2} + \frac{1}{2} \max_{P \in \mathcal{E}_n} Tr(A_{\cl G}^s P). \]
The value $\max_{P \in \mathcal{E}_n} Tr(A_{\cl G}^s P)$ is obtained by considering the eigenvalues of the circulant matrix $A_{\cl G}^s = A_{\cl G}$. These eigenvalues have the form
\[ \lambda_j = \frac{q}{n} - \frac{p}{2n} \omega_n^j - \frac{p}{2n} \omega_n^{(n-1)j} \]
where $\omega_n = e^{2 \pi i/n}$ is the $n$-th root of unity. In particular, the largest eigenvalue of $A_{\cl G}$ is $\frac{q}{2n} + \frac{p}{n}\cos(\pi/n)$. Thus, the value of
\[ \min_{D \in \mathcal{D}_n} Tr(D) \quad \text{subject to} \quad D - A_{\cl G}^s \geq 0, \]
which is equal to
\[ \min_{y \in \mathbb{R}} ny \quad \text{subject to} \quad yI_n - A_{\cl G}^s \geq 0 \]
is given by 
\[ n \left[ \frac{q}{n} + \frac{p}{2n}\cos(\pi/n) \right] = q + p\cos(\pi/n). \]
Finally, repeating the calculations from the proof of Theorem \ref{thm: non-sync value with symmetric distribution} yields the result. 
\end{proof}



 \section{Products of Games}
 There is a great deal of research concerning products of games and especially their behaviour when one does many iterations of a fixed game.\cite{JPY, DSV, BVY}   Many of these results are false for synchronous values of games.
 
 Given two games $G_i= (X_i, O_i, \lambda_i), i=1,2$ their product $G_1 \times G_2$ is the game with input set $X:=X_1 \times X_2$, output set $O:=O_1 \times O_2$ and rule function,
 \[ \lambda: X \times X \times O \times O \to \{ 0,1\} = \bb Z_2,\]
 given by
 \[ \lambda((x_1,x_2),(y_1,y_2), (a_1,a_2),(b_1,b_2)) = \lambda_1(x_1,y_1,a_1,b_1) \lambda_2(x_2,y_2,a_2,b_2),\]
 where the product is in $\bb Z_2$. Thus, they win if and only if $\lambda_1(x_1,y_1,a_1,b_1)=1$ and $\lambda_2(x_2,y_2,a_2,b_2) =1$, that is if and only if they win both games.
 It is customary to write $\lambda= \lambda_1 \times \lambda_2.$
 
 Given prior distributions $\pi_1:X_1 \times X_1 \to [0,1]$ and $\pi_2: X_2 \times X_2 \to [0,1]$ it is easy to see that by defining, 
 \[ \pi: X \times X \to [0,1], \,\, \pi((x_1,x_2),(y_1,y_2)) := \pi_1(x_1,y_1), \pi_2(x_2,y_2),\]
 we obtain a distribution on $X \times X$, which is denoted by $\pi_1 \times \pi_2$.
 
 If $\cl G_i = (G_i, \pi_i)$ denotes the game with distribution $\pi_i$ then we set $\cl G_1 \times \cl G_2 = (G_1 \times G_2, \pi_1 \times \pi_2)$.
 
 These definitions clearly extend to products of more than two games.  Given a game with distribution $\cl G= (G, \pi)$ we let $\cl G^n = (G^n , \pi^n)$ denote the $n$-fold product of a game with itself.

 Here are a few of the results that are known for the values of such games:  
 \begin{enumerate}
 \item (Supermultiplicativity) $\omega_{t}(\cl G \times \cl H) \ge \omega_{t}(\cl G) \omega_{t}(\cl H)$, and examples exist for which the inequality is strict,
 \item $\omega_{t}(\cl G \times \cl H) \le min \{ \omega_{t}(\cl G),  \omega_{t}( \cl H) \}$
 \item $G \times H$ has a perfect t-strategy $\iff$ $G$ and $H$ each have a perfect t-strategy for $t= loc, qa, qc$.
 \item if $\omega_{loc}(\cl G) <1$, then $\omega_t(\cl G^n) \to 0$.
 \end{enumerate}

Thus, when the value is not 1, even though it is possible that $\omega_{loc} (\cl G^n) > \omega_{loc}(\cl G)^n$, we  still have that  it tends to 0.  
 
 The analogues of (1) and (3) were shown to hold for synchronous values in \cite{MPTW}, where an example is also given to show that the inequality can be strict. 
 
 The example below shows that (2) and (4) can fail for synchronous values.

\begin{exam}\label{sayedgame}
Let $\cl G= (G,\pi)$ be the game where Alice's and Bob's question and answer sets are $\{0,1\}$ and let the distribution $\pi$ be given by $\pi_{0,1} = \pi_{1,1} = 1/2$. The players win if their answer pair is $(1,1)$ when asked question pair $(0,1)$. They also win if their answer pair is $(0,1)$ when they receive question pair $(1,1)$. They lose in all other cases. Note that Bob receives $1$ with probability $1$ while Alice receives $0,1$ with equal probability. 

This game has a perfect non-synchronous strategy, namely, for Bob to always return $1$ and for Alice given input $x \in \bb Z_2$ to always return $x+1$. Thus,
\[ \omega_{loc}(\cl G) = \omega_{qc}(\cl G) =1,\]
and consequently,
\[ \omega_{loc}(\cl G^n) = \omega_{qc}(\cl G^n) = 1.\]
\end{exam}

\begin{thm} Let $\cl G=(G, \pi)$ be the game with distribution of Example~\ref{sayedgame}. Then
\[ \omega^s_{loc}(\cl G^n) = \omega^s_{qc}(\cl G^n) = 1 - \frac{1}{2^n}.\]
\end{thm}
\begin{proof}
The synchronous value of this game is at most $1/2$, since on question $(1,1)$ a synchronous strategy will require them to return the same answer and lose.  On the other hand, the deterministic strategy of Alice and Bob always returning $1$ has a value of $1/2$.  
Hence, $\omega_{loc}^s(G) = \omega_{q}^s(G) = \frac{1}{2}$. In terms of traces and projections, this is given by setting $E_{0,1}= E_{1,1} = I$ and $E_{0,0} = E_{1,0} = 0$.

Now for the $n$-fold parallel repetition the questions are pairs $x, y \in \{0,1\}^n$ and the answers are pairs $a,b \in \{ 0, 1 \}^n$. But $\pi^n(x,y) =0$ unless $y=(1,...,1):= 1^n,$
while $\pi(x, 1^n) = \frac{1}{2^n}, \,\, \forall x \in \{ 0, 1 \}^n$.

The only question pair where the synchronous restriction can be enforced is therefore $(1^n,1^n)$, and on this question any synchronous strategy loses as before. Thus,  $\omega^s_{qc}(\cl G^n) \le 1 - \frac{1}{2^n}$. 

On the other hand, consider the deterministic strategy where when the input string is $1^n$ they return $1^n$ but for every other input string $x \ne 1^n$, they return the output string $\overline{x}= x + 1^n$, where addition is in the vector space $\bb Z_2^n$, i.e., each bit of $x$ is flipped.  For every string $x \ne 1^n$ that Alice receives this strategy wins. Hence,  $\omega^s_{loc}(\cl G^n) \ge 1 - \frac{1}{2^n}$. 
Therefore the synchronous value of the parallel repeated game is $\omega_{loc}^s(G^n) = \omega_{qc}^s(G^n)  = 1 - \frac{1}{2^n}$.

Alternatively, this is the strategy that corresponds to choosing PVM's,
\[ E_{1^n,1^n} = E_{x, \overline{x}} = I, \,\, \forall x \ne 1^n,\]
and all other projections equal to 0. 
\end{proof}

Thus, not only does the synchronous value not tend to 0, but it is monotonically increasing.  Also, we have that
\[ \omega^s_t(\cl G^2) > \min \{ \omega_t(\cl G), \omega^s_t(\cl G) \},\]
so that this example violates the synchronous analogues of properties (2) and (4).

Two objections can be raised to this example. The game itself is not synchronous and the distribution is not symmetric.
It is natural to wonder if this pathology persists even when restricting attention to this smaller family of synchronous games with symmetric prior densities.  This is formalized in the following problems.
 
\begin{prob} If $\cl G_i =(G_i, \pi_i), i=1,2$ are symmetric synchronous games with symmetric densities, then 
is $\omega^s_t( \cl G_1 \times \cl G_2) \leq \min(\omega^s_t(\cl G_1),\omega^s_t(\cl G_2))$?
\end{prob}

\begin{prob} If $\cl G$ is a symmetric, synchronous game with symmetric distribution, can $\omega_t^s(\cl G^n)$ be monotone increasing?
\end{prob}

We next return our attention to XOR games. 

First note that the product of two XOR games is not an XOR game. In fact the product is not even a game with binary answers. Our first step is to recall an operation on XOR games, studied in \cite{CSUU}, that unlike the product, produces an XOR game. The {\bf XOR of XOR games} $G_1$ and $G_2$ with densities $\pi_1,\pi_2$ and rule functions $f_1$ and $f_2$, denoted by $G_1\oplus G_2$, is the XOR game $(I_1\times I_2,\{0,1\},\lambda)$ with distribution $\pi_1\times\pi_2$ and rule function $\lambda$ defined so that $\lambda((x_1,x_2),(y_1,y_2),a,b) = 1$ iff $a + b = f_1(x_1,y_1) + f_2(x_2,y_2)$ in $\bb Z_2$. The XOR of more than two games is defined inductively. 

The following result shows why this is an interesting operation on XOR games.

\begin{prop}
Let $\cl G_i= (I_i, \{ 0, 1\}, \lambda_i, \pi_i), i=1,2$ be XOR games with densities and cost matrices $A_{\cl G_i}, i=1,2$. Then the cost matrix of their direct sum satisfies
\[ A_{\cl G_1 \oplus \cl G_2} =  A_{\cl G_1} \otimes A_{\cl G_2}
.\]
\end{prop}

The {\bf bias} of a game with distribution is defined by the formulas
\[ \epsilon_t(\cl G) = 2 \omega_t(\cl G) - 1, \, \, t= loc, q, qc,\]
and corresponds to the probability of winning minus the probability of losing.  Similarly, we have the {\bf synchronous bias},
\[ \epsilon^s_t(\cl G) = 2 \omega^s_t(\cl G) - 1, \, \, t = loc, q, qc.\]
In \cite[Theorem~1]{CSUU} it was proven that the quantum bias of XOR games is multiplicative for the direct sum operations, i.e.,
\[ \epsilon_q( \cl G_1 \oplus \cl G_2) = \epsilon_q(\cl G_1) \epsilon_q(\cl G_2).\]

In what follows we show that this fails for the synchronous bias, even for a family of games that is very well behaved.

\begin{defn}
An XOR game with distribution $\pi$ will be called a {\bf synchronous XOR game}, provided that the game is synchronous, i.e., $f(x,x) =0$, symmetric, $f(x,y) = f(y,x)$ and the distribution is symmetric, $\pi(x,y) = \pi(y,x)$. 
\end{defn}

Note that when $\cl G$ is a synchronous XOR game, we have that the cost matrix $A_{\cl G}= ( (-1)^{f(x,y)} \pi(x,y)) = A_{\cl G}^T$ and hence,
\[ A_{\cl G}^s = A_{\cl G}.\]

In what follows we first show that the perfect parallel repetition does not hold for the synchronous bias of synchronous XOR games. We then identify a subclass of XOR games for which the synchronous value satisfies the perfect parallel repetition. 

Restating Theorem~\ref{thm: semidefinite program for XOR games} in terms of biases yields:

\begin{thm} \label{biasof XOR}
Let $G=(I,\{0,1\},\lambda)$ be an XOR game with $n := |I|$, and suppose $f:I \times I \to \{0,1\}$ is a function satisfying $f(x,y)=a \oplus b$ for all $a,b \in \{0,1\}$. Let $\pi(x,y)$ be a prior distribution on $I$. Then for $\cl G = (G,\pi)$,
\[ \epsilon_{qc}(\cl G) = \epsilon_{q}(\cl G) =  \max_{P \in \mathcal{E}_{2n}} Tr(B_{\cl G} P) \]
and
\[ \epsilon_{qc}^s(\cl G) = \epsilon_{q}^s(\cl G) =  \max_{P \in \mathcal{E}_{n}} Tr(A_{\cl G}^s P). \]
\end{thm}

Fix the question set to be $I = \{1,\ldots,m\}$ and we can equivalently write the above optimization problem for the bias of a synchronous XOR game as the primal-dual semidefinite programs
\begin{equation*}
\openup\jot 
\begin{aligned}[t]
(\mathcal{P})\quad\text{ maximize:}\quad &\ip{A}{P} \\
        \text{subject to:}\quad & \diag(P) = 1,\\
        & P \geq 0,
\end{aligned}
\qquad\qquad 
\begin{aligned}[t]
(\mathcal{D})\quad\text{ minimize:}\quad &\sum_{k=1}^m y_k\\
        \text{subject to:}\quad & \Diag(y) - A \succeq 0,
\end{aligned}
\end{equation*}
where the inner product is the trace inner product, 
\[A:= A_{\cl G}^s = 1/2(\pi(x,y)(-1)^{f(x,y)})+ 1/2(\pi(x,y) (-1)^{f(x,y)})^T, \] 
and $\diag$ is the function that zeros out nondiagonal entries of a matrix, and $\Diag$ of a vector is the matrix where the diagonal entries are the vector entries and nondiagonal entries are zero. This primal-dual satisfies the Slater condition \cite{Sl} and therefore their optimal values are attained and are equal. In fact by complementary slackness if $(P^*,y^*)$ is an optimal solution pair for primal and dual then it holds that $P^*(\Diag(y^*)-A) = 0$. Now if $y'$ is any other optimal dual solution, it holds that $P^*\Diag(y^*-y') = 0$. Since the diagonal entries of $P$ are $1$, this implies that $y' = y^*$. Therefore we get the following lemma

\begin{lemma}
The dual problem $(\mathcal{D})$ has a unique optimal solution. 
\end{lemma}

In the next theorem, we show that the bias of an XOR game for which $\Diag(y^*) \geq A \geq -\Diag(y^*)$, where $y^*$ is the unique dual optimal solution, are multiplicative. That is for any two XOR games with this property, we have $\epsilon^s_q(G_1\oplus G_2) = \epsilon^s_q(G_1)\epsilon^s_q(G_2)$. This in particular includes all XOR games for which the game matrix is positive semidefinite. This is not true for all XOR games as is shown by the next example.

\begin{exam}
Let $\cl G$ be the synchronous XOR game with cost matrix \[A = \begin{bmatrix}\frac{1}{21} & -\frac{3}{21} & -\frac{3}{21}\\
-\frac{3}{21} & \frac{1}{21} & -\frac{3}{21}\\
-\frac{3}{21} & -\frac{3}{21} & \frac{1}{21}
\end{bmatrix}.\]
The pair $P^* = \begin{bmatrix}1 & -\frac{1}{2} & -\frac{1}{2}\\
-\frac{1}{2} & 1 & -\frac{1}{2}\\
-\frac{1}{2} & -\frac{1}{2} & 1\end{bmatrix}$ and $y^* = \begin{bmatrix}\frac{4}{21}\\\frac{4}{21}\\\frac{4}{21}\end{bmatrix}$ are easily seen to be feasible solutions of the primal and dual SDPs and they achieve the same value $\frac{4}{7}$ in the primal and dual problems, respectively. Therefore they are optimal solutions and the optimal value and hence the synchronous  quantum bias of this game is 
\[ \epsilon^s_q(\cl G) =\frac{4}{7}.\] 

Now the cost matrix for the game $\cl G^{\prime}=\cl G \oplus \cl G$ is $A \otimes A$. Therefore the primal-dual problem for $\cl G^{\prime}$ is
\begin{equation*}
\openup\jot 
\begin{aligned}[t]
(\mathcal{P})\quad\text{ maximize:}\quad &\ip{A\otimes A}{W} \\
        \text{subject to:}\quad & \diag(W) = 1,\\
        & W \succeq 0,
\end{aligned}
\qquad\qquad 
\begin{aligned}[t]
(\mathcal{D})\quad\text{ minimize:}\quad &\sum_{k=1}^9 u_k\\
        \text{subject to:}\quad & \Diag(u) - A\otimes A \succeq 0.
\end{aligned}
\end{equation*}
Now from a similar argument like above the pair $W^* = ee^*$ where $e \in \bb C^9$ is the all-one vector and $u = (\frac{5}{21})^2 e$ are optimal solutions for the primal and dual respectively and the optimal value is $(\frac{5}{7})^2$. So we have that 
\[ \epsilon^s_q(\cl G \oplus \cl G)= (\frac{5}{7})^2 > (\frac{4}{7})^2= \epsilon^s_q(\cl G)^2.\]

Note that the unique optimal solution $y^*$ for the dual problem of $G$ does not satisfy the condition
\[\Diag(y^*) \geq A \geq -\Diag(y^*)\] because the eigenvalues of $A$ are $4/21, 4/21, -5/21$.
\end{exam}

\begin{defn}
We call  a synchronous XOR game $\cl G$ and symmetrized cost matrix $A:=A^s_{\cl G}$ {\bf balanced}, if the unique optimal dual solution $y^*$ satisfies
\[\Diag(y^*) \geq A \geq -\Diag(y^*).\]
\end{defn}

Suppose that $\cl G$ is a balanced game and $y^*$ is its unique dual optimal solution. Note that if $y^*(i) \leq 0$ for some question $i$, then the inequalities above imply that $y^*(i) = A(i,i)=0$. Then again since $A + \Diag(y^*)$ is positive semidefinite (and its $i$th diagonal element is $0$), it must be that the $i$th column and row of $A$ are all zeros. Therefore it is true that $\pi(i,j) = \pi(j,i) = 0$ for all questions $j$. Therefore question $i$ is irrelevant and can be removed from the question set of the original game. Thus without loss of generality, we can assume that $y^* > 0$. 

\begin{prop}
Any XOR $\cl G$ game for which $A^s_{\cl G} \geq 0$ is balanced. 
\end{prop}

\begin{thm}
If $\cl G_i, i=1,2$ are balanced XOR games, then
 \[\epsilon_q^s(\cl G_1\oplus \cl G_2) = \epsilon_q^s(\cl G_1)\epsilon_q^s( \cl G_2)\] and $\cl G_1 \oplus \cl G_2$ is balanced.
\end{thm}
\begin{proof}
It is straightforward to see that $\epsilon_q^s(\cl G_1\oplus \cl G_2) \geq \epsilon_q^s(\cl G_1)\epsilon_q^s(\cl G_2)$. So we just prove the reverse inequality $\epsilon_q^s(\cl G_1\oplus \cl G_2) \leq \epsilon_q^s(\cl G_1)\epsilon_q^s(\cl G_2)$.

Let $A_1$ and $A_2$ be the symmetrized cost matrices of $\cl G_1$ and $\cl G_2$, respectively. Then the symmetrized cost matrix of $\cl G = \cl G_1\oplus \cl G_2$ is $A = A_1\otimes A_2$. By assumption the unique optimal dual solutions satisfy $y_1 > 0$ and $y_2 > 0$ and
\begin{gather*}
    -\Diag(y_1) \leq A_1 \leq \Diag(y_1),\\
    -\Diag(y_2) \leq A_2 \leq \Diag(y_2).
\end{gather*}
As we mentioned earlier, without loss of generality, we may assume $\Diag(y_1) > 0$ and $\Diag(y_2) > 0$. So we have
\begin{gather*}
    -I \leq \Diag(y_1)^{1/2} A_1 \Diag(y_1)^{1/2} \leq I,\\
    -I \leq \Diag(y_2)^{1/2} A_2 \Diag(y_2)^{1/2} \leq I.
\end{gather*}
This implies that the operator norm of \[(\Diag(y_1)^{1/2} \otimes \Diag(y_2)^{1/2})(A_1 \otimes A_2)(\Diag(y_1)^{1/2} \otimes \Diag(y_2)^{1/2})\] is at most $1$ and therefore
\[-I\leq (\Diag(y_1)^{1/2} \otimes \Diag(y_2)^{1/2})(A_1 \otimes A_2)(\Diag(y_1)^{1/2} \otimes \Diag(y_2)^{1/2}) \leq I\]
which equivalently can be written as
\begin{gather*}
     -\Diag(y_1) \otimes \Diag(y_2) \leq A_1 \otimes A_2 \leq \Diag(y_1) \otimes \Diag(y_2).
\end{gather*}
Thus $y_1 \otimes y_2$ is a feasible solution of the dual problem of $G_1\oplus G_2$. Therefore the bias of $G_1 \otimes G_2$ is at most $\epsilon_q^s(G_1)\epsilon_q^s(G_2)$. Therefore it must be that $\epsilon_q^s(G_1\oplus G_2) = \epsilon_q^s(G_1)\epsilon_q^s(G_2)$ and $y_1 \otimes y_2$ is the unique dual optimal solution for $G_1 \oplus G_2$. Finally, from the last inequality we derived, the game $G_1\oplus G_2$ is balanced.
\end{proof}

\section{Optimality Conditions}

In this section we derive conditions that a family $E$ of $n$ $k$-PVM's in a tracial C*-algebra $(\cl A, \tau)$ must satisfy in order to give the optimal value of a game. In finite dimensions we 
add the restriction that $E$ is optimal over POVM's
and obtain stronger optimality conditions.

More precisely, given an $n$ input, $k$ output game and distribution, 
$\cl G = (G, \pi)$, and a tracial C*-algebra $(\cl A, \tau)$, 
we seek conditions that a family 
$E:= \{E_{x,a} :  x \in I, a \in O \}$ of $n$ $k$-PVM's 
or $k$-POVM's  
must satisfy in order to maximize the quantity
\[ \phi(E)=  \sum_{(x,y,a,b) \in W} \pi(x,y) \tau(E_{x,a}E_{y,b})
= \sum_{(x,y,a,b) } \pi(x,y) \lambda(x,y,a, b) \tau(E_{x,a}E_{y,b}).
\]
When a family  maximizes $\phi$ over all PVM-families in $\cl A$, we call it
  {\bf optimal for $(\cl A, \tau)$}. 
 We will often be interested in the case $\cl A = M_m$ where it is clear that such a maximum is always attained.
Given a family   
of operators
  $\{ E_{x,a}: x \in I, \, a \in O \}$ 
  as above, and a game with distribution, for each fixed $(x,a) \in I \times O$ we set
  \[ Q_{x,a} = \sum_{\stackrel{y,b}{(x,y,a,b) \in N, y \ne x}} \pi(x,y) E_{y,b} + \sum_{\stackrel{y,b}{(y,x,b,a) \in N, y \ne x}} \pi(y,x) E_{y,b}. \]
  Note that when $\lambda$ is symmetric and the distribution is symmetric, then both sums occurring in the definition of $Q_{x,a}$ are equal. 
  
  \subsection{Optimality over families of PVM's}
  
  We begin with a first derivative condition. 
  
  \begin{prop} 
  \label{pr:vern1deriv}
  Let $(G, \pi) =(I, O, \lambda, \pi)$ be a synchronous game with distribution, let $(\cl A, \tau)$ be a faithful trace of type t, let $\{ E_{x,a} \} \subseteq \cl A$ and let $p(a,b| x,y) = \tau(E_{x,a} E_{y,a})$. If $\{ E_{x,a} \}$ is optimal for $(\cl A, \tau)$, then
  \[ \sum_a E_{x,a} Q_{x,a} = \sum_a Q_{x,a} E_{x,a}, \, \forall x \in I.\]
  \end{prop}
  \begin{proof} Fix $x_0 \in I$.  Let $H = H^* \in \cl A$ and replace the projections $E_{x_0,a}$ by $e^{iHr} E_{x_0,a} e^{-iHr}$, while leaving all the other projections fixed. Let us call these new projections $\{ F_{x,a} \}$ and the density $p_r(a,b|x,y)$ and consider the function,
  \[ f(r) = 1-\omega(\cl G, \pi, p_r)= \sum_{(x,y,a,b) \in N} \pi(x,y) p_r(a,b|x,y).\]
  Note that this function is  a constant except for terms appearing in $Q_{x_0,a}$
  
  Since this smooth function attains its minimum at $r=0$ we must have that 
  \begin{eqnarray}
  0=f^{\prime}(0) & = & i\sum_{\stackrel{a,y,b}{(x_0,y,a,b) \in N, y \ne x_0} }\pi(x_0,y) \tau(H E_{x_0,a}E_{y,b} - E_{x_0,a}H E_{y,b}) \nonumber \\ 
  & + & i \sum_{\stackrel{a,y,b}{(y,x_0,b,a) \in N, y \ne x_0}} \pi(y,x_0) \tau(E_{y,b}H E_{x_0,a} - E_{y,b} E_{x_0,a}H) \nonumber \\ 
  & = & i \sum_{\stackrel{a,y,b}{(x_0, y,a,b) \in N, y \ne x_0}} \pi(x_0,y) \tau(H(E_{x_0,a}E_{y,b} - E_{y,b}E_{x_0,a})) \nonumber \\ 
  & + & i \sum_{\stackrel{a,y,b}{(y,x_0,b,a) \in N, y \ne x_0}} \pi(y,x_0) \tau(H(E_{x_0,a}E_{y,b} - E_{y,b}E_{x_0,a})) \nonumber \\
  & = & i \tau(H(\sum_a E_{x_0,a} Q_{x_0,a} - Q_{x_0,a}E_{x,a})). \nonumber
  \end{eqnarray}
Since this is true for every $H=H^*$ and $\tau$ is faithful, we have that
\[ \sum_a E_{x_0,a} Q_{x_0,a} - Q_{x_0,a} E_{x_0,a} =0,\]
and the result follows.
\end{proof}

\begin{remark} This proof is adapted from \cite{DPP} where a similar idea was used to prove that for the graph correlation function, if a set of projections $\{ P_x : x \in V \}$ minimized the correlation for a graph $(V, E)$, then necessarily each $P_x$ commuted with the sum of the projections over all vertices adjacent to $x$.
\end{remark}

\begin{remark} 
We show what this result says about the CHSH game, with uniform distribution.  Recall that this game has $I=O = \bb Z_2$ and the rules are that to win $a+b = xy$ where the arithmetic is in the field $\bb Z_2$.
Computation shows that
\[ Q_{0,0} = E_{1,1}, \, Q_{0,1} = E_{1,0}, \, Q_{1,0} = E_{0,1}, \, Q_{1,1} = E_{0,0}.\]
Thus, the above result tells us that for an optimum strategy,
\[ E_{0,0}E_{1,1} + E_{0,1} E_{1,0} = E_{1,1} E_{0,0} + E_{1,0}E_{0,1}.\]
Setting $P= E_{0,0}, Q= E_{1,0}$, this equation becomes
\[ P(I-Q) + (I-P)Q = (I-Q) P + Q(I-P) \implies PQ = QP.\]
Thus, an optimal synchronous strategy for this game is an abelian strategy, which shows that
\[ \omega^s_{qc}(CHSH) = \omega^s_{loc}(CHSH),\]
and we know that this latter value is the supremum over all deterministic strategies where Alice and Bob must use the {\it same} function $f:I \to O$.
It is well-known that among these four functions, the optimal is for Alice and Bob to always return 0, i.e.,  $f(x) =0, \forall x$, which has a value of $3/4$.
Thus, this game has no quantum advantage when we restrict to synchronous strategies.
\end{remark} 

\begin{remark} If we have $\{ E_{x,a} \}$ optimal as above and we set
$$\Omega_x = \sum_a E_{x,a} Q_{x,a}$$ then the result is equivalent to $\Omega_x = \Omega_x^*, \, \forall x$ and it is also equivalent to the condition that $E_{x,a} Q_{x,a} E_{x,b} = E_{x,a} Q_{x,b} E_{x,b}, \,\forall a,b,x.$
\end{remark}

One difficulty with C*-algebras is that they might contain few projections, for example the C*-algebra of continuous functions on $[0,1]$ only contains the two trivial projections.  However, von Neumann algebras are always generated by their projections. Given any C*-algebra and faithful trace $(\cl A, \tau)$ after we take the GNS representation, we may always look at the  tracial von Neumann algebra generated by the image.  Thus, insisting that $\cl A$ be a von Neumann algebra does not impose an undue restriction.

\begin{lemma} Let $(\cl A, \tau)$ be a von Neumann algebra with a faithful trace $\tau$, let $E$ be a projection and let $H=H^*$.  If for every projection $P \le E$ we have that $\tau(PH) \ge 0,$ then $EHE \ge 0$.
\end{lemma}

  \begin{prop} 
  \label{pr:gameOpt}
  Let $(G, \pi) =(I, O, \lambda, \pi)$ be a synchronous game with distribution, let $(\cl A, \tau)$ be a von Neumann algebra with a faithful trace. If $\{ E_{x,a} \}$ is optimal for $(\cl A, \tau)$, then

  \begin{equation} 
 E_{x,a} Q_{x,b} E_{x,a} + \delta_x^b E_{x,a}
 \le E_{x,a} Q_{x,a}E_{x,a} + \delta_x^a E_{x,a},
 \end{equation} 
 where $\delta_x^a= \pi(x,x)\lambda(x,x,a,a)$.
  \end{prop}
  
  \begin{proof} 
 Fix an $x_0$ and a pair $a_0 \ne b_0$, and a projection, $P \le E_{x_0,a_0}$. If we replace the family $\{ E_{x,a} \}$ by the family $\{ F_{x,a} \}$ defined by
\begin{itemize} 
\item
$F_{x,a} = E_{x,a}, \forall x \ne x_0$,
\item $F_{x_0,c} = E_{x_0,c}, \forall c \ne a_0, b_0$
\item $F_{x_0, a_0} = E_{x_0,a_0} -P$,
\item $F_{x_0,b_0} = E_{x_0, b_0} +P$,
\end{itemize}
then the value $\phi(F)$ of this new family of projections must be smaller than $\phi(E)$. Computing $\phi(E) - \phi(F)$ and applying the above lemma yields the result.
\end{proof}

\begin{remark} 
\label{rem:optcolor}
It is instructive to see what these results tell us in the case of the graph k-colouring game with uniform distribution. If a set of projections $\{ E_{x,a} \}$ is optimal for this game and we write $y \sim x$ to indicate that vertices $x,y$ are adjacent, then $Q_{x,a} = 2\sum_{y \sim x} E_{y,a}$ and the first derivative result tells us that for each $a$,
\[ \sum_{a} E_{x,a} Q_{y,a} = \sum_{a} Q_{y,a} E_{x,a}.\]
The second result implies that
\[ E_{x,a} ( \sum_{y \sim x} E_{y,b} ) E_{x,a} \le E_{x,a}( \sum_{y \sim x} E_{y,a} ) E_{x,a},\]
since $\delta_x^a= \delta_x^b$.  Summing this inequality over all $b, b \ne a$ yields,

\[  d_x E_{x,a}  \le k E_{x,a} (\sum_{y\sim x} E_{y,a}) E_{x,a},\]
where $d_x$ is the degree of the vertex $x$.
\end{remark}

\begin{remark}
The  necessary condition for a family to be optimal for  
$(\mathcal A, \tau)$  in   Proposition \ref{pr:vern1deriv}
comes from  $f'(0) = 0$.
An additional  necessary condition for optimality 
comes from analyzing $f''(0) \leq 0$, which 
we did  successfully and we found inequalities on the 
$E_{x,a} Q_{x,b} E_{x,a}$ which are equivalent.
Comparing  these inequalities  to those
in Proposition  \ref{pr:gameOpt}   
yields, when $\mathcal A= M_n$, 
  that the conditions in Proposition \ref{pr:vern1deriv}
    and   Proposition  \ref{pr:gameOpt}
  implies $f''(0) \leq 0$.
  This is unexpected, since one is derived by calculus and the other from permuting projections.
We omit the proof, since as just noted it does not give a new
optimality condition and the proof is not short.
\end{remark}


\def\brx{{x_0}}


\subsection{Optimizing over  POVM's}
This subsection makes different assumptions than the previous one.
There we studied optimizers over PVM's.
Here, we consider optimizing trace functionals over the bigger set of POVM's, which might well produce a
higher maximum. 

Note that if $\{ E_{x,a} \}$ are only POVM's, then setting
\[ p(a,b|x,y) = \tau(E_{x,a}E_{y,b}),\]
does define a density in $C_{qc}$, see Lemma 5.2 of \cite{Ru}. But it will not necessarily be a synchronous density.  In fact, assuming that $\tau$ is a faithful trace, we will have that the density is synchronous if and only if $\tau(E_{x,a}E_{x,b}) =0$ for $a \not = b$ which is equivalent to $E_{x,a}E_{y,b} =0$. On the other hand the fact that $\sum_a E_{x,a} = I$ and $E_{x,a}E_{x,b} =0$ implies that each $E_{x,a}$ is a projection.  Thus, the set of densities that can be obtained in this fashion is strictly larger than the synchronous densities, but it is also known to be smaller than the set of all densities in $C_{qc}$.  For more details on this set of densities see \cite{Ru}.

We write $\omega_{povm}(G, \pi)$ for the value of a game over the set of densities obtained as traces of POVM's. 

Throughout this section, we. will also require the extra restriction that  $\lambda(x,x,a,a)=1, \forall x,a$.

With these assumptions we can
 use semidefinite programing theory, 
and so get powerful optimality conditions which easily  imply the conclusions
of Proposition  \ref{pr:gameOpt}
and Proposition  \ref{pr:vern1deriv}
restricted to finite dimensions.

\begin{prop}
\label{prop:coptsum}

Let $(\cl G, \pi) =(I, O, \lambda, \pi)$ be a synchronous 
  game with distribution such that  
  $\lambda(x,x,a,a)=1$ for all $a,x$ and let  $(\cl A, \tau) =(M_n, tr_n)$ be the $n \times n$ matrices with their unique normalized trace.
 An optimizing  POVM for $(\cl A, \tau)$ 
  which is a PVM must satisfy

 \begin{enumerate}
 \item
 \label{it:bigineq}
 $
\Omega_x -  Q_{x,b}   \ge 0 \ \ for \ all \ b
$
\\

 \item
  $
  (  \Omega_x -  Q_{x,b}  ) E_{x,b} =0= 
 E_{x,b}    ( \Omega_x -  Q_{x,b}  ) 
$ \ \ for \ all \ b.
\end{enumerate} 
Suppose the max value of the game occurs with a finite dimensional strategy which  is a synchronous strategy. Then the hypotheses of this proposition  apply;
so (\ref{it:bigineq}) and (2) must both hold. 
\end{prop}

\begin{remark}
An immediate consequence of this 
for the graph coloring problem is
 gotten by summing on 
$b= 1, \dots, k$.
 We obtain
\begin{equation}
\label{eq:povmColor}
\Omega_x \geq  \frac {d_x} {k }  I  
\qquad   \qed
\end{equation}
This condition compressed by $E_{x,a}$
is the same as the (weaker) one in Remark \ref{rem:optcolor}.
So one naturally thinks of it as just saying
that the (block) diagonal entries  of \eqref{eq:povmColor}
are all  positive semidefinite.

It is easily seen that
the conclusions of Proposition \ref{pr:gameOpt}
 are immediate consequences 
of
Proposition \ref{prop:coptsum}\eqref{it:bigineq}  
 (which requires  stronger hypothesis on $\lambda$). 
\end{remark}


\subsubsection{POVM proofs}



It will be very useful to sort $\phi(E)$  according to dependence 
on a particular point $\brx$.  Let $E_x$ denote the POVM \
$ E_x:= \{ E_{x,1}, \cdots, E_{x,k} \}$.

\begin{lemma}
\label{lem:phisort} 
Fix  $\tau$ and fix $\brx  $. 
Then
\begin{equation}
    \phi(E)= no(\brx) 
    + \mu(E_\brx )  
  +  \sum_{a,b} \lambda(\brx ,\brx, a,b) \; 
    \tau (E_{\brx,a} E_{\brx,b}).
\end{equation}
$no(\brx)$ has no dependence on $\brx$ and
$\mu(E_x) := 
\tau ( \sum_{a}  E_{x,a} Q_{x,a} ) .$
\end{lemma}

\begin{proof}
This is a straightfoward decomposition of the sum
defining $\phi$.
\end{proof}

We are assuming  that the diagonal of 
$\lambda(\brx, \brx, a ,a) =1 $, that $\lambda$ is synchronous, and that
$E_x$ is a PVM.
Thus  by Lemma \ref{lem:phisort}  we get
   $ \phi(E)= no(\brx) + \mu(E_{\brx,a})$
%
%
Therefore an optimal family for $(\cl A, \tau)$ at a fixed $\brx$
must maximize 
$$
\;  \tau ( \sum_{a=1}^{k-1} 
 E_{\brx,a} Q_{\brx,a} ) 
 +
 \ \tau ( \;  I-( E_{\brx,1} + \dots E_{\brx, k-1} )
  Q_{x,k}) \;  
 $$
%
subject to 
\\
$$E_{\brx,a} \ge 0 \ \  all \ a = 1 \dots, k-1
\quad 
and 
\quad   I-( E_{\brx,1} + \dots E_{\brx, k-1} ) \geq 0 $$
with $E_{\brx,k} $ being set equal to the last expression.
This  is a Semi Definite Program (SDP)
over a domain with interior
whose
 dual  SDP is 
 
\def\RR{{\mathbb R}}

\begin{equation}
\label{eq:min}
 \min_R\tau(R^{kk})
\end{equation}
subject to

\begin{enumerate}
\item
$R \in \RR^{nk \times nk} $ is PSD, with $R$ partitioned as
$R= : \begin{pmatrix}
R^{11} & \dots & R^{1k} \\ 
\vdots & \ddots & \vdots\\
R^{k1} & \dots & R^{kk} 
\end{pmatrix}.
$
\item 
$R^{aa} - R^{kk} = -  Q_{\brx,a}  + Q_{\brx,k}$  
\end{enumerate}
The off diagonal terms of $R$ are irrelevant and we 
ignore them from now on.

Here  the  standard Primal -Dual Optimality Conditions,
see \cite{AHO},
are

\begin{lemma}
\label{lem:psdcopt}
 If the POVM $E_x$ and the dual optimizer 
$R$ exist (i.e. the optimum is  achieved) for the core problem, 
then they satisfy  for all $a$:
\begin{enumerate}
\item
$R^{aa} \ge 0 $
\\
\item
$R^{aa} E_{\brx,a} =0= E_{\brx,a} R^{aa} $\\
\item 
$R^{aa} - R^{kk} = Q_{\brx,a} - Q_{\brx,k}$
\ for all $a$.
\end{enumerate}
\end{lemma}

\noindent
{\it Proof of Proposition \ref{prop:coptsum}}
\ \
Proposition \ref{prop:coptsum} follows from 
Lemma \ref{lem:psdcopt} as we now see.
First observe that
$R^{aa} - R^{bb} = - Q_{\brx,a} + Q_{\brx,b}$,
because     
\begin{align*}
  R^{aa}  - R^{bb} 
&= R^{aa} - R^{kk} - (R^{bb} - R^{kk}). \\
 \end{align*}
From this we get
$\Omega -   Q^b  =   R^{bb} \ge 0$ for all $b$,
because
\begin{align*} 
   \sum_a E_{\brx,a} ( Q_{\brx,a} - Q_{\brx,b} )
&=   - \sum_a  E_{\brx,a} R^{aa}   +  \sum_a  E_{\brx,a} R^{bb},\\
\Omega - Q^b &=  +  \sum_a  E_{\brx,a} R^{bb} 
= R^{bb}. \quad  
\end{align*} 

Letting $\omega_{povm}(G,\pi)$ denote the max value of the game over all densities of the form $\tau(E_{x,a} E_{y,b})$ for POVMs $\{E_{x,a}\}$ in a finite dimensional von Neumann algebra with a trace $\tau$, the last assertion of the proposition
has main assumption which implies
 $$
  \omega_q^s(G, \pi) \leq  \omega_{povm}(G, \pi)  \leq    \omega_{qc}(G, \pi) =\omega_{q}(G, \pi);
  $$
the second  inequality following from Lemma 5.2 of \cite{Ru}.
Thus our maximizing PVM strategy  is a POVM
maximizer which
 amounts to the demanding hypothesis of
Proposition  \ref{prop:coptsum}. \qed

  \newpage
  
  \tableofcontents

\end{document}